\documentclass[12pt]{article}
\usepackage[margin=1in]{geometry}
\usepackage{setspace}
\usepackage{amsmath,amssymb,amsfonts,amsthm,mathtools}
\usepackage{bm}
\usepackage{graphicx}
\usepackage{subcaption}
\usepackage{float}
\usepackage{booktabs}
\usepackage{multirow}
\usepackage{array}
\usepackage{enumitem}
\usepackage{xcolor}
\usepackage{hyperref}
\usepackage[numbers,sort&compress]{natbib}
\usepackage{url}
\usepackage{mdframed}

\hypersetup{colorlinks=true,linkcolor=blue,citecolor=blue,urlcolor=blue}
\newtheorem{theorem}{Theorem}[section]
\newtheorem{proposition}[theorem]{Proposition}

\theoremstyle{definition}

\theoremstyle{remark}
\newtheorem{remark}[theorem]{Remark}

\newcommand{\E}{\mathbb{E}}
\newcommand{\Pbb}{\mathbb{P}}

\newcommand{\F}{\mathcal{F}}
\newcommand{\G}{\mathcal{G}}
\newcommand{\A}{\mathcal{A}}
\newcommand{\given}{\,|\,}
\newcommand{\ind}{\mathbf{1}}
\newcommand{\argmin}{\operatorname*{arg\,min}}

\newmdenv[backgroundcolor=blue!4,linecolor=blue!30,linewidth=1pt,
          roundcorner=4pt,innertopmargin=7pt,innerbottommargin=7pt,
          skipabove=8pt,skipbelow=8pt]{clinbox}

\title{\textbf{Optimal Stopping in Sequential Clinical Prediction}}

\author{Foo Hui-Mean~ and ~Yuan-chin Ivan Chang\\
\small Institute of Statistical Science, Academia Sinica, Taipei, Taiwan\\
\small \texttt{ycchang@as.edu.tw}}

\date{\today}

\begin{document}
\maketitle

\begin{abstract}
Most clinical prediction studies are developed from retrospective cohorts and
reported as if all patient information were observed at once.  In practice,
clinicians face a more consequential question: \emph{when is there already
enough information to stop testing and act?}  A later stage can produce a
better-looking model and still fail to justify the added delay, burden, or
invasiveness of further workup.

We formulate sequential clinical prediction as an \emph{optimal-stopping}
problem under staged information, and illustrate the framework across four
retrospective clinical datasets.  The preferred stopping stage differed
substantially by setting: sometimes fuller information justified waiting,
whereas in other cases early or intermediate action was preferable.  The key
object is the patient-specific conditional risk trajectory: forward martingale
structure represents coherent risk updating across stages, while
reverse-martingale ideas describe information loss when a richer predictor is
replaced by a simpler score.  The results demonstrate that the best-performing
model is not always the best stage for clinical decision-making.

\end{abstract}

\noindent\textbf{Keywords:} optimal stopping; sequential clinical prediction; staged testing; risk updating; decision analysis; information loss; cancer; heart disease; diabetes; ICU mortality; eICU; Statistics in Medicine.


\section{Introduction}
\label{sec:intro}

Most clinical classification models are built from retrospective,
population-based data and summarized in a cross-sectional way: one feature set,
one fitted model, and one set of performance measures.  That tradition is
useful, but it does not match how medical decisions are actually made.  In
practice, a clinician begins with information that is immediately available,
adds tests in stages, and repeatedly asks a practical question that is usually
left outside standard prediction papers: \emph{should we stop now, or is the
next test worth waiting for?}

For many medical applications, that stopping decision is the central scientific
question.  A later stage can have a higher AUC and still fail to be the best
stage at which to act.  Conversely, an early stage may already support a
reasonable decision even when more information could still improve prediction
slightly.  The clinically relevant issue is therefore not simply whether
prediction improves, but whether the improvement is large enough to justify the
added burden, delay, invasiveness, or resource use of further assessment.

This paper studies staged clinical prediction from that perspective.  The
organizing viewpoint is \emph{optimal stopping}: sequential prediction is
useful because it supports a decision about when enough information has been
collected.  The statistical machinery explains how to formalize that question.
In particular, the paper is not simply comparing a sequence of population
models with increasingly large feature sets.  The stage-specific models are
indeed estimated from retrospective cohorts, but the scientific object of
interest is the \emph{individual patient's conditional risk trajectory} across
staged information.  In other words, the framework asks how the risk for the
patient currently under evaluation should update as the workup unfolds, and
whether that update is large enough to justify moving to the next stage.

That distinction deserves to be explicit because otherwise the empirical
results could be dismissed as the trivial claim that later feature sets often
predict better than earlier ones.  What is nontrivial here is the linkage
between subject-specific risk updating and the decision of whether to continue
testing.  Forward martingale structure describes whether that patient-level risk
updating is coherent as information accumulates.  Reverse-martingale structure
addresses a different question: what is lost when richer information is
deliberately compressed into a simpler score or earlier-stage summary.  The
clinical purpose is then to determine the preferred stopping stage under stated
loss and testing-cost assumptions.

The martingale formulation of Bayesian updating \citep{doob1953,williams1991}
has rarely been used as an explicit organizing principle for clinical
prediction; related ideas appear in Bayesian clinical monitoring
\citep{berry1985,spiegelhalter1994} and sequential analysis
\citep{wald1945,armitage1975,siegmund1985}, although those settings focus on
trial monitoring rather than individual patient workups.  Decision curve
analysis \citep{vickers2006} and proper scoring rules \citep{gneiting2007}
provide closely related tools for threshold-based evaluation; our aim is to
place risk updating, score simplification, and decision timing within a common
staged framework.  The same stopping logic extends naturally to
\emph{treatment selection}---when the clinician must choose among therapies,
the question becomes whether one more stage changes the preferred treatment
enough to justify its burden---a connection discussed in the appendix.

Against this background, the paper pursues four related aims.  First, it
formalizes the stopping decision via a Bellman recursion, so that an additional
test is ordered only when its expected reduction in decision loss exceeds its
added burden.  Second, it represents the patient's evolving risk as a
martingale under increasing clinical information, giving a principled benchmark
for whether stagewise risk estimates are updating coherently.  Third, it
quantifies the information lost when a richer stage-specific predictor is
compressed into a simpler clinical score, using a projection-loss
decomposition linked to the reverse-martingale structure.  Fourth, it
illustrates all three ideas across four retrospective medical studies---Breast
Cancer Wisconsin, Cleveland Heart Disease, Pima Diabetes, and the eICU
Collaborative Research Database Demo---representing a range of clinical
settings and stopping conclusions.

The remainder of the paper is organized as follows.
Section~\ref{sec:methodology} develops the methodology: the statistical
framework (martingale and reverse-martingale structure, Bayesian decision rule,
Bellman recursion), the empirical estimands and patient-level bridge
diagnostics, and the datasets and implementation.
Section~\ref{sec:results} presents the numerical results for all four studies.
Sections~\ref{sec:discussion} and~\ref{sec:conclusion} discuss the findings,
address limitations, and summarize the main messages.  Proofs of the main
theoretical results are collected in Appendix~\ref{app:proofs} at the end of
this paper; appendix study results, sensitivity analyses, calibration
tables, and code details are provided in the appendix.

\section{Methodology}
\label{sec:methodology}

\subsection{Statistical Framework}
\label{sec:framework}

We view each patient workup as a sequence of clinically ordered information
states.  At stage $t$, the clinician has observed everything in $\F_t$ and
forms a current estimate of disease risk.  The notation below provides the
formal version of that idea.

\subsubsection{Clinical state, filtration, and posterior belief}
\label{sec:setup}

Let $(\Omega,\mathcal{F},\Pbb)$ be a probability space, and let $D\in\{0,1\}$
denote a latent binary outcome (disease present, or adverse event).  Let
$(\F_t)_{t\ge 0}$ be an increasing filtration, $\F_0\subseteq\F_1\subseteq
\cdots$, where $\F_t$ represents all clinical information available by stage
$t$ (history, examination, laboratory results, imaging, etc.).  The key object
is the posterior risk process
\begin{equation}
    X_t := \E[D \given \F_t] = \Pbb(D=1 \given \F_t).
    \label{eq:posterior}
\end{equation}
This is the clinician's current belief that the patient has the target
condition, given all evidence accumulated through stage $t$.
Table~\ref{tab:notation} summarizes the main mathematical objects and their
clinical interpretations.

\begin{table}[htbp]
\centering
\caption{Main mathematical objects and their clinical interpretations.}
\label{tab:notation}
\renewcommand{\arraystretch}{1.25}
\begin{tabular}{p{2.8cm}p{4.6cm}p{5.0cm}}
\toprule
\textbf{Object} & \textbf{Statistical meaning} & \textbf{Clinical meaning} \\
\midrule
$D\in\{0,1\}$ & binary outcome & disease / adverse-event status \\
$\F_t$ & increasing $\sigma$-algebra & accumulated clinical information \\
$X_t = \E[D|\F_t]$ & conditional risk (martingale) & updated posterior belief \\
$\G_t$ & decreasing $\sigma$-algebra & coarsened information / score layer \\
$Y_t = \E[D|\G_t]$ & projected risk (reverse mart.) & risk-score-based belief \\
$a_t^*$ & Bayes-optimal action & treat / do not treat \\
$c^*$ & decision threshold & $c_{\mathrm{FP}}/(c_{\mathrm{FP}}+c_{\mathrm{FN}})$ \\
$\tau^*$ & optimal stopping time & first stage where acting is optimal \\
\bottomrule
\end{tabular}
\end{table}

\subsubsection{Coherent risk updating under increasing information}
\label{sec:martingale}

\begin{proposition}[Posterior risk is a martingale]
\label{prop:martingale}
Assume $D\in L^1(\Pbb)$.  The process $(X_t,\F_t)_{t\ge 0}$ is a
martingale:
\(
    \E[X_{t+1} \given \F_t] \;=\; X_t \quad \text{a.s.}
\)
If $D\in[0,1]$, then $(X_t)$ is bounded and converges almost surely and in
$L^1$ to $X_\infty = \E[D\given\F_\infty]$, where $\F_\infty =
\bigvee_{t\ge 0}\F_t$.
\end{proposition}

\begin{proof}[Proof sketch]
This is an immediate consequence of iterated expectation and standard
martingale convergence results; a full proof is given in
Appendix~\ref{proof:martingale}.
\end{proof}

\begin{remark}
{(Clinical interpretation.)}  The martingale property means there is no
systematic tendency for belief to drift further in the same direction after a
test result is incorporated.  It does not depend on the model class used to
estimate $X_t$.
\end{remark}

\begin{remark}
The martingale property is an idealization.  In practice, estimated risks
$\hat{X}_t$ are not exact conditional expectations; Section~\ref{sec:bridge}
describes how to quantify and diagnose deviations.
\end{remark}

\subsubsection{Information loss when rich data are reduced to scores}
\label{sec:revmart}

Let $(\G_t)_{t\ge 0}$ be a \emph{decreasing} filtration,
$\G_0\supseteq\G_1\supseteq\cdots$, representing the coarsened information
encoded in a clinical risk score $S_t$ (e.g., SOFA score, CURB-65, Framingham
risk score) so that $\G_t=\sigma(S_t)$.  We assume throughout that
$\G_t\subseteq\F_t$ for all $t$: the score at each stage is a coarsening of
the full clinical information available at that stage.  Set
\begin{equation}
    Y_t := \E[D \given \G_t].
    \label{eq:projected}
\end{equation}

\begin{proposition}[Projected risk is a reverse martingale]
\label{prop:revmart}
Assume $D\in L^1$.  Then $(Y_t,\G_t)_{t\ge 0}$ satisfies
$\E[Y_t\given\G_{t+1}]=Y_{t+1}$ a.s., so it is a reverse martingale.  It
converges a.s.\ and in $L^1$ to $Y_\infty=\E[D\given\G_\infty]$, where
$\G_\infty=\bigcap_{t\ge 0}\G_t$.
\end{proposition}

\begin{proof}[Proof sketch]
The result follows from iterated expectation under a decreasing information
sequence; a full proof is given in Appendix~\ref{proof:revmart}.
\end{proof}

\begin{proposition}[Projection-loss decomposition]
\label{prop:projection_loss}
Assume $D\in L^2$ (which holds automatically since $D\in\{0,1\}$) and
$\G_t\subseteq\F_t$.  Then
\[
    \E[(D-Y_t)^2] - \E[(D-X_t)^2] \;=\; \E[(X_t-Y_t)^2] \;\ge\; 0.
\]
\end{proposition}

\begin{proof}[Proof sketch]
This is the standard $L^2$ projection identity for conditional expectation; a
full proof is given in Appendix~\ref{proof:projection}.
\end{proof}

The quantity $\E[(X_t-Y_t)^2]$ is the \emph{projection loss}: the mean
squared error incurred by summarising the full posterior through the score $S_t$.

\subsubsection{Treatment as a Bayes decision}
\label{sec:decision}

Let $\A=\{0,1\}$ (treat/not treat) with asymmetric costs: $L(1,0)=
c_{\mathrm{FP}}$ (cost of treating a healthy patient) and $L(0,1)=
c_{\mathrm{FN}}$ (cost of missing a diseased patient), with zero loss for
correct decisions.

\begin{proposition}[Threshold Bayes rule]
\label{prop:threshold}
The Bayes-optimal action at stage $t$ is
\[
    a_t^* = \ind\{X_t > c^*\}, \qquad
    c^* = \frac{c_{\mathrm{FP}}}{c_{\mathrm{FP}} + c_{\mathrm{FN}}}.
\]
\end{proposition}

\begin{proof}[Proof sketch]
Compare the expected loss of treating with the expected loss of not treating;
the threshold expression follows directly.  A full proof is given in
Appendix~\ref{proof:threshold}.
\end{proof}

\begin{remark}[Threshold and net benefit]
The threshold $c^*$ is the risk level above which treatment becomes preferable
under the stated loss ratio, equivalent to the treatment threshold in decision
curve analysis \citep{vickers2006}.
\end{remark}

\subsubsection{Optimal stopping: when to act}
\label{sec:stopping}

The clinician faces a sequential decision: at each stage $t$, either act
immediately using current belief $X_t$, or pay a test cost $c_t>0$ and wait
until stage $t+1$.  Let $V_t$ denote the minimum attainable total expected
loss from stage $t$ onward.

\begin{theorem}[Dynamic programming for optimal clinical stopping]
\label{thm:dp}
Let $T$ be the finite horizon (total number of stages).  Set $V_T=\ell(X_T)$
as the boundary condition, and for $t=T-1,\ldots,0$ define recursively
\begin{equation}
    V_t = \min\!\left\{
        \underbrace{\inf_{a\in\A}\E[L(a,D)\given\F_t]}_{\text{cost of acting at }t},\;
        \underbrace{c_t + \E[V_{t+1}\given\F_t]}_{\text{cost of continuing}}
    \right\}.
    \label{eq:bellman}
\end{equation}
The optimal stopping time is $\tau^*=\min\{t : V_t = \inf_{a}\E[L(a,D)|\F_t]\}$.
Under the threshold rule of Proposition~\ref{prop:threshold}, the acting cost at
stage $t$ equals
\[
    \ell(X_t) := c_{\mathrm{FN}}\,X_t\,\ind\{X_t\le c^*\}
               + c_{\mathrm{FP}}\,(1-X_t)\,\ind\{X_t > c^*\}.
\]
(When $X_t\le c^*$ the decision is not to treat, incurring expected false-negative
cost $c_{\mathrm{FN}}\,X_t$; when $X_t>c^*$ the decision is to treat, incurring
expected false-positive cost $c_{\mathrm{FP}}(1-X_t)$.)
Stopping is optimal at stage $t$ whenever $\ell(X_t)\le c_t+\E[V_{t+1}|\F_t]$.
\end{theorem}

\begin{proof}[Proof sketch]
This is a standard Bellman recursion for a finite-horizon stopping problem;
see \citet{chow1971os}.  A fuller derivation is provided in
Appendix~\ref{proof:dp}.
\end{proof}

\begin{remark}[Retrospective approximation of the Bellman value]
\label{rem:dp_prospective}
In prospective clinical use, $V_{t+1}$ in~\eqref{eq:bellman} is unknown at
decision time $t$, as it depends on future observations not yet available.
In the retrospective studies below, we approximate the continuation cost
$c_t + \E[V_{t+1}|\F_t]$ by the observed total expected loss at stage $t+1$,
evaluated on the held-out test set.  This proxy treats the next stage as the
final stage, so it is conservative: it underestimates the true continuation
value by ignoring the possibility of early stopping at even later stages
($t+2, t+3,\ldots$).  The reported total cost at each stage is therefore an
upper bound on the true Bellman value, making our stopping recommendations
conservative.
\end{remark}

\begin{remark}
\textbf{Clinical interpretation.}  The physician should order an additional
test only when the expected reduction in decision loss---i.e., the expected
improvement in diagnostic accuracy---exceeds the cost of the test (time,
patient discomfort, resources).  When this condition fails, the optimal
policy is to act on current evidence.
\end{remark}

\subsection{Empirical Estimands and Diagnostics}
\label{sec:bridge}

In retrospective clinical data, the true conditional expectations $X_t$ and
$Y_t$ are not directly observed but are estimated by stagewise prediction
models.  Let $\hat{X}_t = f_t(\F_t)$ denote a calibrated risk estimate at
stage $t$ and $\hat{Y}_t = g_t(\G_t)$ a score-based estimate.  The empirical
goal is not merely to show that later models can look better.  Rather, it is
to assess whether $(\hat{X}_t)$ behaves approximately martingale-like as a
patient-level risk trajectory, to quantify projection loss
$\E[(\hat{X}_t - \hat{Y}_t)^2]$, and to identify the stage at which it is
preferable to stop testing and act.

\begin{remark}[Three-part conceptual separation]
\label{rem:separation}
It is useful to separate three distinct operations that are sometimes
conflated in sequential prediction studies:
\begin{enumerate}[label=\textbf{(\roman*)},leftmargin=2.2em,itemsep=1pt]
\item \textbf{Population learning} --- estimating stage-specific models
  $\hat{X}_t$ from retrospective cohort data.
\item \textbf{Individual updating} --- applying those models to track the
  conditional risk for the \emph{patient currently under evaluation} as
  additional information arrives.
\item \textbf{Decision timing} --- determining, for that patient, whether one
  more stage of testing is expected to reduce decision loss by more than its
  added burden, delay, or invasiveness.
\end{enumerate}
Forward martingale structure pertains to levels (ii) and (iii): it
characterizes how the individual patient's risk estimate should update
coherently as new evidence arrives, and it supports the stopping decision.
Reverse-martingale structure describes a separate operation at level (i):
what is lost when a richer stage-$t$ predictor is deliberately compressed
into a simpler clinical score or earlier-stage summary.  Population-level
metrics (AUC, Brier score) summarize level (i) and supply the raw materials
for levels (ii) and (iii), but they do not by themselves answer the stopping
question at level (iii).
\end{remark}

\begin{proposition}[Decision regret under posterior compression]
\label{prop:regret}
Suppose the decision loss $\ell(\cdot)$ is Lipschitz with constant $L$.
Then
\[
    \E[\ell(\hat{Y}_t)] - \E[\ell(\hat{X}_t)] \;\le\; L\,\E[|\hat{X}_t - \hat{Y}_t|].
\]
\end{proposition}

\begin{proof}[Proof sketch]
Apply the Lipschitz bound $|\ell(x)-\ell(y)|\le L|x-y|$ pointwise to
$(\hat{X}_t,\hat{Y}_t)$, then take expectations.  For the threshold loss of
Proposition~\ref{prop:threshold}, the Lipschitz constant is
$L=\max(c_{\mathrm{FP}},c_{\mathrm{FN}})$, and the bound is tight; see the
Appendix Section~\ref{sec:S1} for the full proof.
\end{proof}

Table~\ref{tab:theory-to-empirical} maps the theoretical quantities to their
empirical counterparts.

\begin{table}[htbp]
\centering
\caption{Mapping from theoretical framework to empirical estimands.}
\label{tab:theory-to-empirical}
\renewcommand{\arraystretch}{1.2}
\begin{tabular}{p{5cm}p{4.5cm}p{4.5cm}}
\toprule
\textbf{Theory} & \textbf{Empirical proxy} & \textbf{Diagnostic metric} \\
\midrule
$X_t = \E[D|\F_t]$ & $\hat{X}_t$ (logistic regression) & AUC, Brier score, calibration \\
$Y_t = \E[D|\G_t]$ & $\hat{Y}_t$ (score/compressed model) & Prob-MSE vs.\ $\hat{X}_t$ \\
$\E[X_{t+1}-X_t|\F_t]=0$ & conditional mean of $\hat{X}_{t+1}-\hat{X}_t$ & drift $\approx 0$? \\
$\ell(\hat{X}_t)+c_t$ & decision loss + test cost & total cost at each stage \\
$\tau^*$ & $\argmin_t[\text{total cost}(t)]$ & optimal stopping stage \\
\bottomrule
\end{tabular}
\end{table}

\paragraph{Empirical martingale diagnostic.}
For each pair of successive stages, compute the increment
$\Delta_t = \hat{X}_{t+1} - \hat{X}_t$ and estimate the conditional mean
$\E[\Delta_t | \hat{X}_t]$ by binning $\hat{X}_t$ into quantile groups.  Under
an exact martingale, $\E[\Delta_t|\hat{X}_t]=0$ in every bin.  We report the
weighted mean conditional drift $M_t = \sum_b w_b\,\bar{\Delta}_{t,b}$ and the
weighted mean squared conditional drift $S_t = \sum_b w_b\,\bar{\Delta}_{t,b}^2$
as diagnostic summaries.

\paragraph{Patient-level bridge quantities.}
The framework in Theorem~\ref{thm:dp} is patient-specific, but the empirical
implementation relies on population-average stagewise comparisons.  To
partially bridge this gap, we report two patient-level quantities for each
stage transition.
\begin{itemize}[leftmargin=1.8em,itemsep=2pt]
\item[] \textbf{Proportion of patients with a stable decision:}
  $\widehat{\pi}_{t} = n_{\mathrm{te}}^{-1}
  \sum_{i=1}^{n_{\mathrm{te}}}\ind\{\hat{a}_{t,i}^* = \hat{a}_{t+1,i}^*\}$,
  where $\hat{a}_{t,i}^* = \ind\{\hat{X}_{t,i} > c^*\}$.  A high value means
  that moving to the next stage rarely changes the clinical recommendation;
  a low value means many patients are near the decision threshold and could
  be reclassified.
\item[] \textbf{Mean distance from the decision threshold:}
  $\bar{d}_t = n_{\mathrm{te}}^{-1}\sum_i |\hat{X}_{t,i} - c^*|$.  A large
  value indicates that most patients are comfortably above or below the
  threshold (confident decisions), while a small value indicates widespread
  uncertainty about the optimal action.
\end{itemize}
These two quantities are computed for each study after the datasets and
implementation are described; Table~\ref{tab:patient_level} in
Section~\ref{sec:results} reports the values for all four studies together
with their interpretation.

\subsection{Data and Implementation}

\subsubsection{Datasets}
\label{sec:datasets}

Four publicly available retrospective datasets are used.  Table~\ref{tab:datasets}
summarizes the corresponding clinical settings and staged information
structures.

\begin{table}[htbp]
\centering
\caption{Datasets used in the four retrospective studies.
BCW = Breast Cancer Wisconsin.}
\label{tab:datasets}
\renewcommand{\arraystretch}{1.25}
\footnotesize
\setlength\tabcolsep{4pt}
\setlength\hfuzz{3pt}
\begin{tabular}{p{3.2cm}rp{2.8cm}p{5.5cm}p{1.7cm}}
\toprule
\textbf{Dataset} & \textbf{$n$} & \textbf{Outcome} &
\textbf{Staged structure} & \textbf{Access} \\
\midrule
BCW (Study~1) & 569 & Malignant/benign &
  3 blocks: mean / SE / worst-value & UCI (free) \\
Cleveland (Study~2) & 303 & Heart disease &
  3 stages: resting / exercise test / imaging & UCI (free) \\
{Pima:Diabetes} (Study~3) & 768 & Diabetes onset &
  3 stages: screen / +pedigree / full 8 features & UCI (free) \\
{eICU:Demo} (Study~4) & 2{,}359 & In-hospital death &
  3 stages: 24-h vitals / +labs / +demographics & PhysioNet (free) \\
\bottomrule
\end{tabular}
\end{table}

The first three datasets are well-established benchmark collections from the
UCI Machine Learning Repository.  The Breast Cancer Wisconsin (Diagnostic)
dataset \citep{wolberg1994} comprises $n=569$ patients described by 30
morphological features derived from fine-needle aspirate images; the features
are naturally grouped into three blocks representing the mean, standard error,
and worst-value summaries of the same ten cell-nucleus measurements.  The
Cleveland Heart Disease dataset \citep{detrano1989} contains $n=303$ patients
with binary coronary-artery disease status, staged from resting clinical
variables (age, sex, chest-pain type, resting ECG, cholesterol) through
exercise-test results to imaging and invasive findings (fluoroscopy vessel
count, thalassemia).  The Pima Indians Diabetes dataset \citep{smith1988}
contains $n=768$ women of Pima Indian heritage with binary diabetes-onset
status, staged from accessible metabolic measurements (glucose, BMI, age)
through reproductive history and family pedigree to a full panel of clinical
variables.  All three datasets have been widely used in the machine-learning
and clinical-prediction literature, and their staging structures are
analyst-constructed approximations to plausible clinical workup sequences.

\paragraph{Study~4: eICU Collaborative Research Database Demo.}
The eICU Collaborative Research Database Demo \citep{pollard2018eicu} is a
freely available subset of the eICU-CRD, containing $2{,}520$ ICU admissions
from 2014--2015 (downloaded from PhysioNet without credentials).  After
restricting to patients with at least one valid vital-sign measurement in the
first 24 hours, the analytic cohort comprised $n=2{,}359$ admissions, with an
in-hospital mortality rate of 8.6\%.  The three information stages are:
\begin{align*}
\F_1 &= \text{8 vital-sign means (first 24 h): heart rate, SpO}_2\text{, temperature, respiration,}\\
     &\quad\quad \text{systolic/diastolic/mean BP, EtCO}_2\text{,}\\
\F_2 &= \F_1 \cup \text{16 first-available lab values: K, Na, Cl, BUN, creatinine, glucose, Hgb, Hct,}\\
     &\quad\quad \text{WBC, platelets, bicarbonate, albumin, AST, ALT, calcium, magnesium,}\\
\F_3 &= \F_2 \cup \text{5 static variables: age, sex, ICU unit type, admission height, weight.}
\end{align*}
\noindent\textbf{Note on staging convention.}
Demographics ($\F_3$: age, sex, ICU unit type, height, weight) are typically
recorded at admission, not temporally after laboratory results.  Study~4
therefore uses \emph{analytic staging} ordered by incremental predictive
value rather than literal chronology.  A sensitivity ordering with
demographics at $\F_1$ is examined in Appendix Section~\ref{sec:S4}.


\subsubsection{Statistical models}
\label{sec:models}

At each information stage $t$, the primary model for $\hat{X}_t$ is a
standardized logistic regression with $\ell_2$ regularization and
\texttt{liblinear} solver, with features centered and scaled within each
training fold.  The projection model $\hat{Y}_t$ is a principal-component-based
logistic regression (1 or 3 components) representing simplified score-like
summaries of the full stage-specific information.

\subsubsection{Evaluation metrics}
\label{sec:metrics}

\textbf{Prediction} was assessed by AUC (discrimination), Brier score, and log-loss.
\textbf{Projection loss} was measured by
\[
\widehat{\mathrm{PL}}_t =
n_{\mathrm{te}}^{-1}\sum_{i=1}^{n_{\mathrm{te}}}
(\hat{X}_{t,i} - \hat{Y}_{t,i})^2.
\]
\textbf{Decision metrics} included sensitivity, specificity, accuracy, and the empirical decision loss
\[
\hat{\ell}_t =
(c_{\mathrm{FP}}\cdot\mathrm{FP} + c_{\mathrm{FN}}\cdot\mathrm{FN})/n_{\mathrm{te}}
\]
evaluated at threshold $c^*$. The \textbf{stopping metric} was the total cost,
\[
\hat{\ell}_t + \mathrm{cumulative\ test\ cost}(t).
\]
\textbf{Drift} was summarized by the weighted mean conditional drift $M_t$ and weighted mean squared conditional drift $S_t$ using 10 quantile bins. \textbf{Calibration} was assessed by the calibration slope and intercept from logistic recalibration of the observed outcome on the log-odds of the fitted probability in the test set, averaged over 1{,}000 splits; a slope near 1 and an intercept near 0 indicate adequate calibration, with results reported in Appendix Section~\ref{sec:S5}.

\subsubsection{Validation protocol}
\label{sec:validation}

Each study uses 1{,}000 stratified random train/test splits (70\%/30\%).
All analyses use Python 3.12 with \texttt{scikit-learn}, \texttt{numpy},
\texttt{pandas}, and \texttt{matplotlib}.  Code details are provided in
Appendix Section~\ref{sec:S3}.

\subsubsection{Decision parameters}
\label{sec:decision_params}

For Studies~1--3 we set $c_{\mathrm{FP}}=1$ and $c_{\mathrm{FN}}=5$,
giving the Bayes threshold $c^*=1/6\approx 0.167$.  Clinically, this means
that a false negative is treated as five times as costly as a false positive,
so intervention is preferred once the estimated risk exceeds about 17\%.

The cumulative stage costs are illustrative rather than literal monetary
reimbursement values.  They are intended to summarize the added burden of
moving from one stage of workup to the next, including resource use, patient
inconvenience, invasiveness, and delay.  We use $(0, 0.01, 0.03)$ for
Studies~1,~3, and~4, where the later stages represent modestly more intensive
assessment, and $(0, 0.02, 0.06)$ for Study~2, where the pathway moves from a
resting baseline assessment to exercise testing and then to imaging/invasive
information.  Among the four applications, Study~4 is the only setting with a
rare adverse outcome and an ICU triage interpretation, so its decision threshold
is specified separately.  This is not a methodological exception, but an
application-specific choice of loss ratio within the same optimal-stopping
framework.  For Study~4, we set $c_{\mathrm{FN}}=10$ (with $c_{\mathrm{FP}}=1$),
giving $c^*=1/11\approx 0.091$, to reflect that a missed in-hospital death is
treated as ten times more costly than an unnecessary escalation of care in this
ICU setting.  The stopping recommendations should therefore be read as examples
of how conclusions depend on the clinical trade-off, not as fixed
cost-effectiveness claims.  To assess robustness, Appendix Section~\ref{sec:S4}
reports the preferred stopping stage for Studies~2 and~4 under a range of
alternative stage-cost schedules; the main conclusions hold across all
schedules examined.

\section{Numerical Results}
\label{sec:results}

Projection-loss analyses and martingale drift diagnostics omitted from the main
paper are collected for all four studies in Appendix Section~\ref{sec:S2}.

Before presenting each study in turn, Table~\ref{tab:patient_level} provides
a cross-study summary of the two patient-level bridge quantities defined in
Section~\ref{sec:bridge}: the proportion of patients whose threshold decision
is stable across a stage transition ($\hat{\pi}_t$) and the mean distance of
predicted risk from the decision threshold ($\bar{d}_t$).  These numbers
foreshadow the key finding of each study and allow direct comparison across
settings.

\begin{table}[htbp]
\centering
\caption{Patient-level bridge quantities across all four studies (mean over
200 repeated 70/30 stratified splits).  $\hat{\pi}_{t}$ = proportion of test
patients whose threshold decision is unchanged when moving to the next
information stage.  $\bar{d}_t$ = mean distance of the predicted risk from
the decision threshold $c^*$.  High $\hat{\pi}$ indicates the next stage
rarely changes management; low $\bar{d}$ indicates many patients are near
the action boundary.  We use 200 rather than 1{,}000 splits because
$\hat{\pi}_t$ and $\bar{d}_t$ are proportions and distances whose Monte Carlo
standard error is negligible at this count; results were indistinguishable
when re-run at 500 splits.}
\label{tab:patient_level}
\renewcommand{\arraystretch}{1.2}
\begin{tabular}{lrrrrrr}
\toprule
 & \multicolumn{2}{c}{Prop.\ stable $\hat{\pi}$} &
   \multicolumn{3}{c}{Mean dist.\ from $c^*$: $\bar{d}$} \\
\cmidrule(lr){2-3}\cmidrule(lr){4-6}
Study ($c^*$) & $\F_1\to\F_2$ & $\F_2\to\F_3$ &
  $\F_1$ & $\F_2$ & $\F_3$ \\
\midrule
BCW ($c^*\!=\!1/6$)       & 0.977 & 0.965 & 0.501 & 0.509 & 0.531 \\
Cleveland ($c^*\!=\!1/6$) & 0.899 & 0.880 & 0.343 & 0.346 & 0.352 \\
Pima ($c^*\!=\!1/6$)      & 0.952 & 0.980 & 0.309 & 0.310 & 0.309 \\
eICU ($c^*\!\approx\!0.091$) & 0.984 & 0.892 & 0.350 & 0.315 & 0.292 \\
\bottomrule
\end{tabular}
\end{table}

For Pima Diabetes (Study~3), 98\% of patients already have a stable decision
at $\F_1$: moving to $\F_2$ changes the recommended action for only 2 patients
in 100, which directly justifies stopping at $\F_1$ under any small test cost.
For eICU (Study~4), 11\% of patients have their decision reversed when moving
from $\F_2$ to $\F_3$, producing the large observed drop in decision loss.
For Cleveland (Study~2), both transitions are the least stable
($\approx88$--$90$\%), indicating that more patients are near the decision
boundary, yet the gain from the final stage is still too small to justify its
added cost.  The mean distance $\bar{d}$ reflects the difficulty of the
classification task: BCW posteriors are far from the threshold
($\bar{d}\approx0.51$), consistent with near-perfect separation; Pima and
eICU posteriors sit much closer ($\bar{d}\approx0.31$--$0.35$), meaning many
patients are genuinely uncertain under the stated cost ratio.

\subsection{Study 1: Breast Cancer Wisconsin --- When Waiting for Full Review Is Worthwhile}
\label{sec:bcw}

\subsubsection{Dataset and staged structure}
\label{sec:bcw_data}

The Breast Cancer Wisconsin (Diagnostic) dataset
\citep{wolberg1994} contains $n=569$ patients (357 benign, 212 malignant)
described by 30 numeric features derived from digitised images of fine-needle
aspirates.  These features are naturally partitioned into three blocks based on
the standard summary statistics of the same ten cell-nucleus measurements:
\begin{align*}
\F_1 &= \text{10 mean features (radius, texture, perimeter, area, \ldots)},\\
\F_2 &= \F_1 \cup \text{10 standard-error features},\\
\F_3 &= \text{all 30 features}~(\F_2 \cup \text{10 worst-value features}).
\end{align*}
Mean features summarize typical cell morphology, standard-error features capture within-slide variability, and worst-value features flag the most abnormal cell---information available only after complete pathology review.

\subsubsection{Stagewise predictive performance}
\label{sec:bcw_prediction}

Table~\ref{tab:bcw_stage} reports the mean predictive performance over
1{,}000 repeated train/test splits.  All four metrics improve monotonically as
the pathology review becomes more complete.  AUC increased from $0.985$ at the
mean-feature stage (\(\F_1\)) to $0.987$ after adding within-slide variability
features (\(\F_2\)) and to $0.995$ with the full pathology summary
(\(\F_3\)); Brier score fell from $0.045$ to $0.038$ and $0.020$.  The small
SDs show that these gains were stable across repeated splits.

The gain from $\F_1$ to $\F_3$ is clinically meaningful: better separation directly supports more reliable downstream decisions.

\begin{table}[htbp]
\centering
\caption{Stagewise predictive performance over 1{,}000 repeated train/test
splits (Breast Cancer Wisconsin, Study~1).  Mean and SD across
repetitions.  The table reports cohort-level summaries of an underlying
subject-specific risk-updating process under staged pathology information.}
\label{tab:bcw_stage}
\renewcommand{\arraystretch}{1.15}
\begin{tabular}{llrr}
\toprule
Stage & Metric & Mean & SD \\
\midrule
$\F_1$ (mean features)   & AUC      & 0.9848 & 0.0063 \\
$\F_1$                   & Brier    & 0.0452 & 0.0093 \\
$\F_1$                   & Accuracy & 0.9364 & 0.0160 \\
$\F_1$                   & Log-loss & 0.1493 & 0.0284 \\
\addlinespace
$\F_2$ (mean + SE)       & AUC      & 0.9867 & 0.0059 \\
$\F_2$                   & Brier    & 0.0382 & 0.0088 \\
$\F_2$                   & Accuracy & 0.9534 & 0.0151 \\
$\F_2$                   & Log-loss & 0.1333 & 0.0281 \\
\addlinespace
$\F_3$ (all 30 features) & AUC      & 0.9947 & 0.0041 \\
$\F_3$                   & Brier    & 0.0203 & 0.0063 \\
$\F_3$                   & Accuracy & 0.9767 & 0.0106 \\
$\F_3$                   & Log-loss & 0.0786 & 0.0242 \\
\bottomrule
\end{tabular}
\end{table}

\begin{figure}[htbp]
\centering
\begin{subfigure}[t]{0.48\textwidth}
\includegraphics[width=\textwidth]{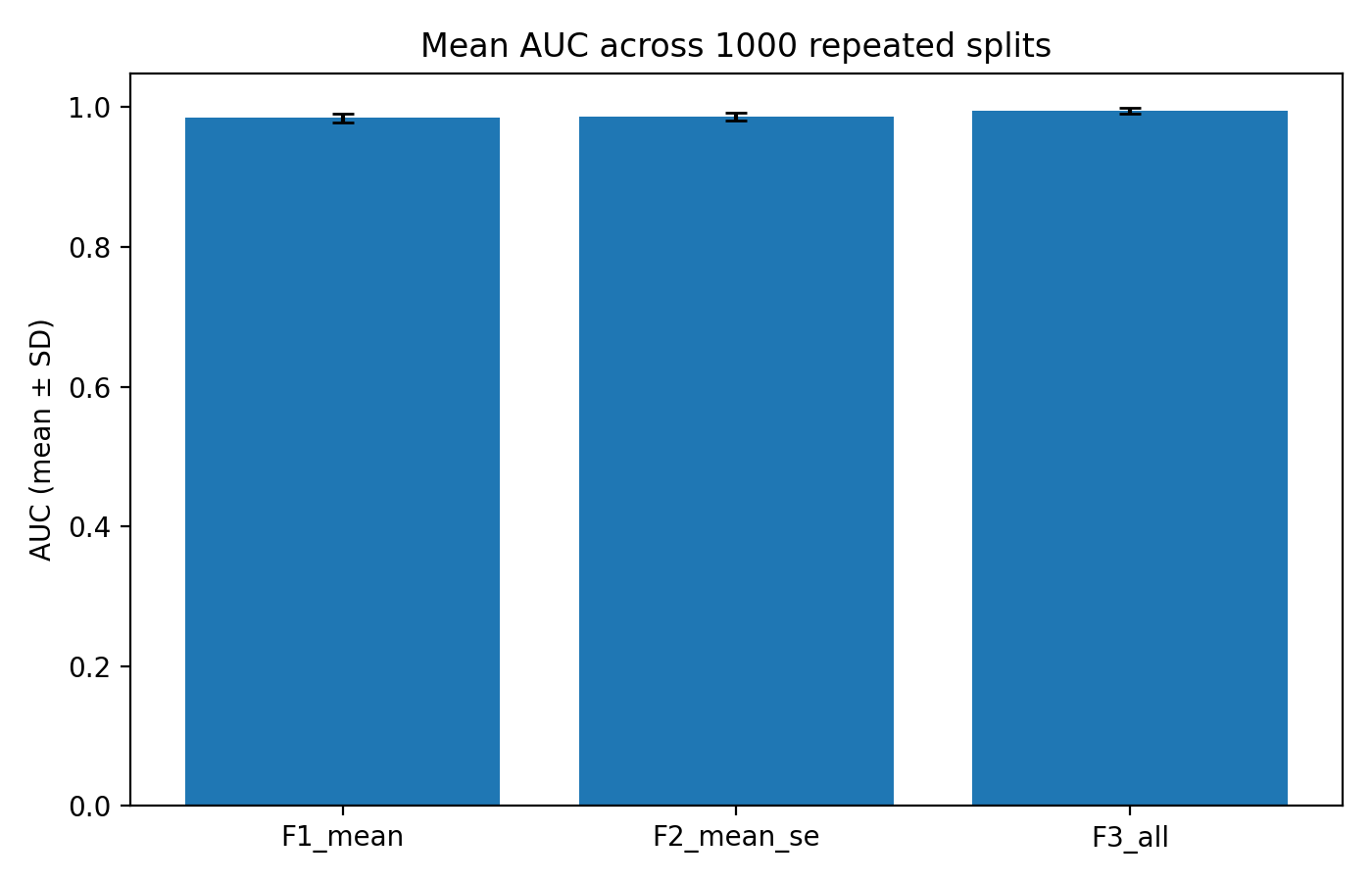}
\caption{Mean AUC at each stage.}
\label{fig:bcw_auc}
\end{subfigure}
\hfill
\begin{subfigure}[t]{0.48\textwidth}
\includegraphics[width=\textwidth]{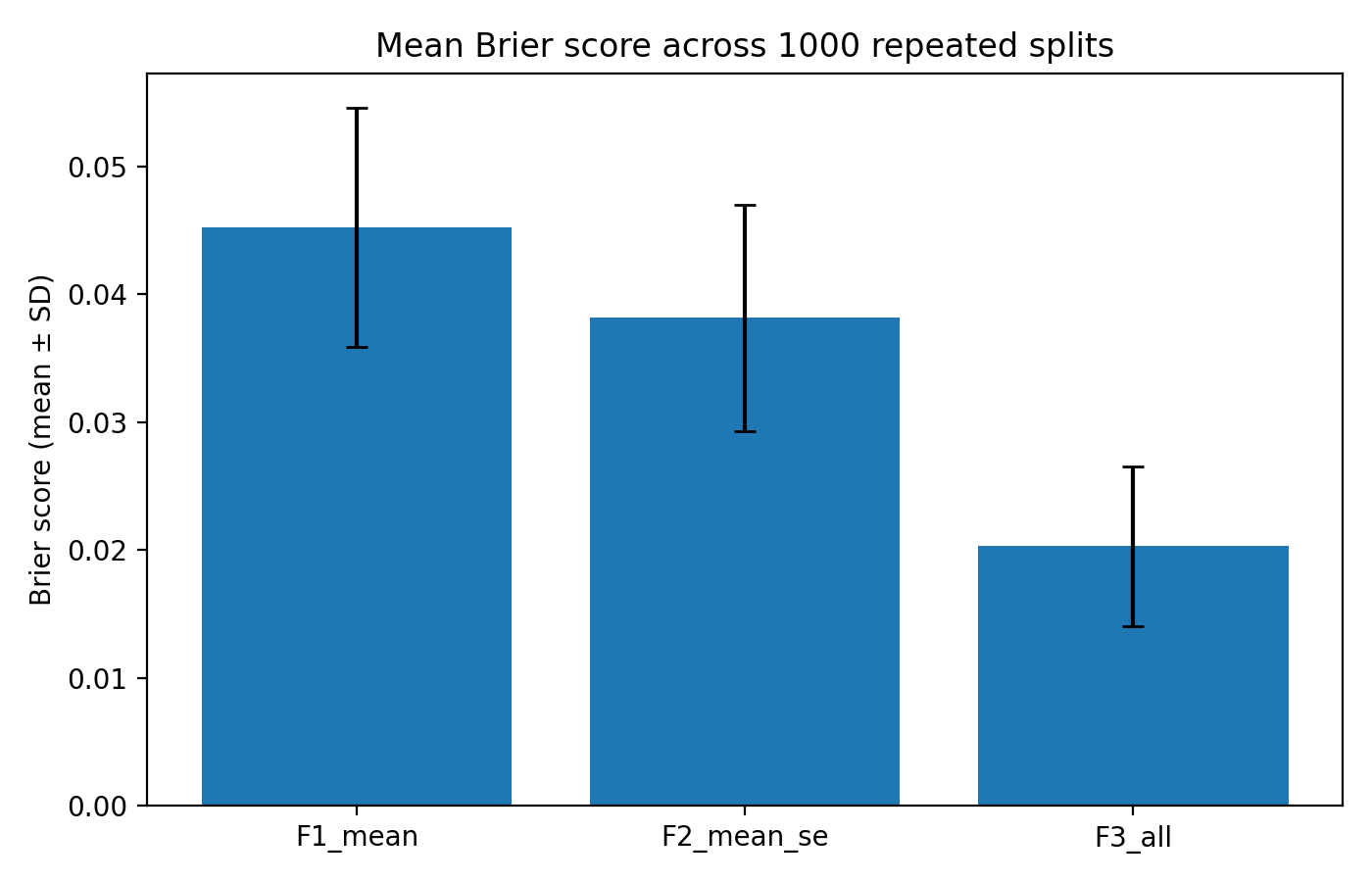}
\caption{Mean Brier score at each stage (lower is better).}
\label{fig:bcw_brier}
\end{subfigure}
\caption{Breast Cancer Wisconsin (Study~1): stagewise predictive performance
over 1{,}000 repeated experiments.  Both discrimination (AUC) and calibration
(Brier) improve monotonically as additional feature blocks are incorporated.
Read clinically, the figure summarizes how individual patient risk estimates
become more informative as the pathology review becomes more complete.}
\label{fig:bcw_performance}
\end{figure}

\subsubsection{Decision performance and optimal stopping}
\label{sec:bcw_stopping}

Under the threshold $c^*=1/6$, expected decision loss fell from $0.154$ at
$\F_1$ to $0.142$ at $\F_2$ and $0.086$ at $\F_3$.  Adding the illustrative
cumulative stage costs $(0.00, 0.01, 0.03)$ yields total costs
$(0.154, 0.152, 0.116)$, so the full-information stage $\F_3$ remains
preferred under this schedule (Figure~\ref{fig:bcw_stopping}).  Sensitivity
exceeded $96\%$ at all stages, while specificity improved markedly from $0.863$
to $0.935$.

From a medical decision-making perspective, these results suggest that the
additional pathology review is clinically worthwhile: the improvement in
classification is large enough to offset the added effort of obtaining the
fuller information.  Here the subject-specific risk updates are large enough to
change management in a clinically meaningful way, not merely to improve a
population-average performance summary.

\begin{table}[htbp]
\centering
\caption{Decision performance and stopping analysis (Study~1).
Total cost = decision loss + cumulative test cost.  Reading down the table,
decision loss falls enough from $\F_1$ to $\F_3$ to more than offset the added
stage cost, so waiting for the full pathology review is preferred under this
illustrative trade-off.  The cohort averages shown here summarize the
patient-level consequences of acting at each stage.}
\label{tab:bcw_decision}
\renewcommand{\arraystretch}{1.15}
\begin{tabular}{lrrrrrr}
\toprule
Stage & Sensitivity & Specificity & Accuracy & Dec.\ loss & Test cost & Total \\
\midrule
$\F_1$ & 0.9638 & 0.8631 & 0.9007 & 0.1535 & 0.00 & 0.1535 \\
$\F_2$ & 0.9644 & 0.8803 & 0.9118 & 0.1416 & 0.01 & 0.1516 \\
$\F_3$ & 0.9756 & 0.9354 & 0.9505 & 0.0860 & 0.03 & 0.1160 \\
\bottomrule
\end{tabular}
\end{table}

Table~\ref{tab:bcw_decision} can be read literally as follows: if the goal is
to decide whether to act after a partial versus fuller review of the aspirate,
then the reduction in wrong decisions from $\F_1$ to $\F_3$ is larger than the
extra stage cost built into the example.  Figure~\ref{fig:bcw_stopping}
therefore highlights a setting in which waiting for fuller information is
clinically worthwhile.

\begin{figure}[htbp]
\centering
\includegraphics[width=0.58\textwidth]{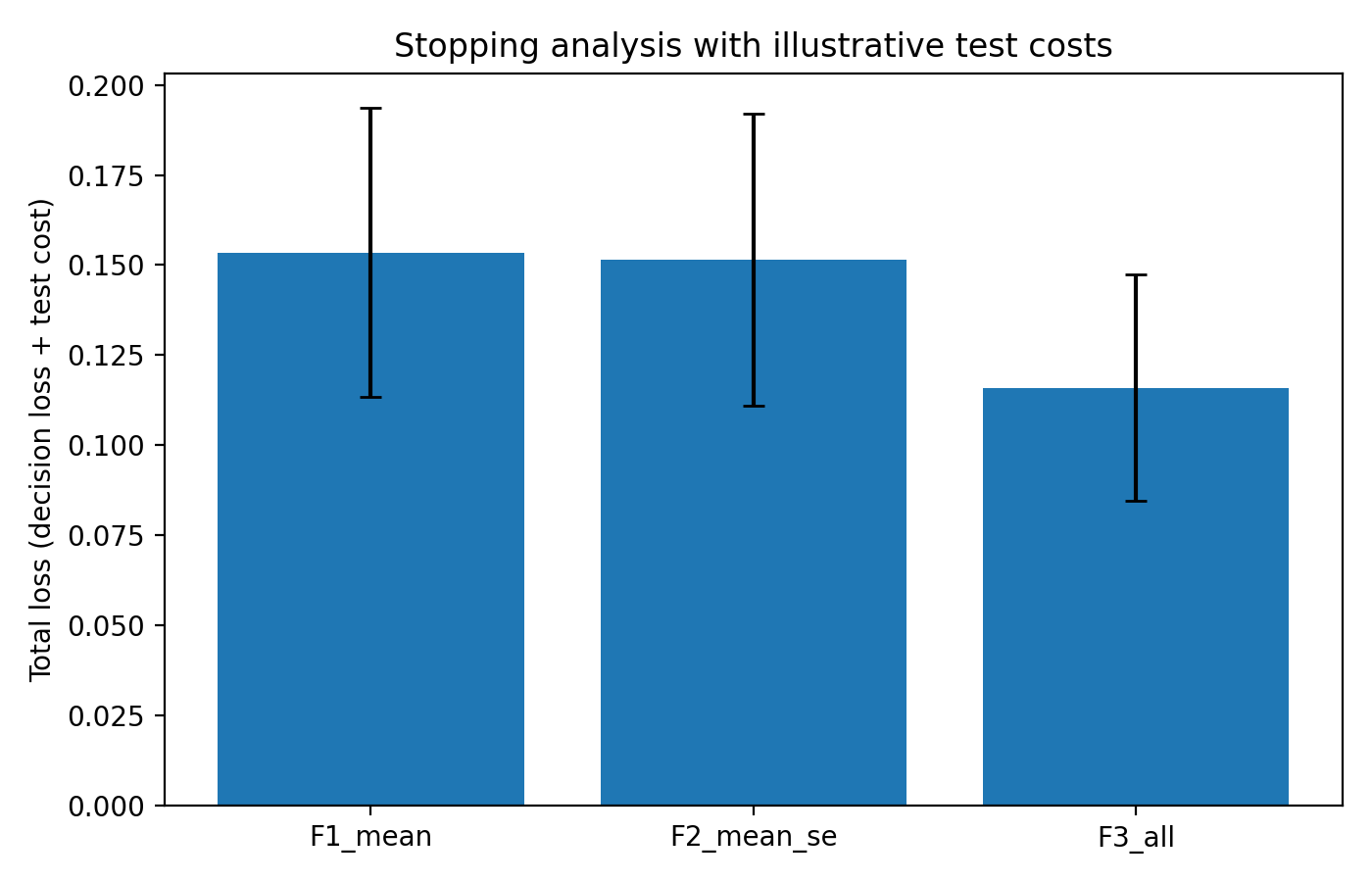}
\caption{Study~1: total expected cost (decision loss + cumulative test cost) by
stage.  The full-information stage $\F_3$ minimizes total cost in this
example.  The literal reading is simple: the extra pathology detail changes the
decision often enough, and in the right direction often enough, that its gain
exceeds the added burden of waiting for it.}
\label{fig:bcw_stopping}
\end{figure}

\subsection{Study 2: Cleveland Heart Disease --- Stopping Before the Final Workup}
\label{sec:cleveland}

\subsubsection{Dataset and staged structure}
\label{sec:cleveland_data}

The Cleveland Heart Disease dataset \citep{detrano1989} contains $n=303$
patients (164 disease-free, 139 with heart disease) described by 13 features
from a standard cardiac workup.  The three information stages follow the natural clinical progression of a cardiac workup:
\begin{align*}
\F_1 &= \{\text{age, sex, chest-pain type, resting BP, cholesterol, fasting blood sugar, resting ECG}\},\\
\F_2 &= \F_1 \cup \{\text{maximum heart rate, exercise-induced angina, ST depression, ST slope}\},\\
\F_3 &= \text{all 13 features}~(\F_2 \cup \{\text{major vessels by fluoroscopy, thalassemia}\}).
\end{align*}
$\F_1$ is the resting baseline, $\F_2$ adds the exercise-test result, and $\F_3$ adds imaging/invasive information.

\subsubsection{Stagewise predictive performance}
\label{sec:cleveland_pred}

Table~\ref{tab:cleveland_stage} reports mean performance over 1{,}000
train/test splits.  AUC improved monotonically from $0.820$ at the resting
baseline stage (\(\F_1\)) to $0.861$ after exercise testing (\(\F_2\)) and to
$0.898$ after adding imaging/invasive information (\(\F_3\)); Brier score fell
from $0.174$ to $0.153$ and $0.128$.  Sensitivity decreases slightly from
$\F_1$ to $\F_3$ (from $0.941$ to $0.917$), whereas specificity increases
markedly (from $0.352$ to $0.586$).

Discrimination and calibration improve monotonically across stages, but the relevant question is whether the extra improvement is large enough to justify the added burden of later-stage workup.

\begin{table}[htbp]
\centering
\caption{Stagewise predictive performance over 1{,}000 repeated train/test
splits (Cleveland Heart Disease, Study~2).  Mean and SD across repetitions.
These cohort summaries are interpreted as evidence about subject-specific risk
updating across a staged cardiac workup.}
\label{tab:cleveland_stage}
\renewcommand{\arraystretch}{1.15}
\begin{tabular}{llrr}
\toprule
Stage & Metric & Mean & SD \\
\midrule
$\F_1$ (resting baseline)   & AUC      & 0.820 & 0.036 \\
                            & Brier    & 0.174 & 0.018 \\
                            & Accuracy & 0.750 & 0.039 \\
                            & Log-loss & 0.532 & 0.051 \\
\addlinespace
$\F_2$ ($\F_1{}+{}$exercise test) & AUC & 0.861 & 0.032 \\
                            & Brier    & 0.153 & 0.020 \\
                            & Accuracy & 0.775 & 0.038 \\
                            & Log-loss & 0.470 & 0.059 \\
\addlinespace
$\F_3$ ($\F_2{}+{}$imaging) & AUC      & 0.898 & 0.028 \\
                            & Brier    & 0.128 & 0.021 \\
                            & Accuracy & 0.829 & 0.033 \\
                            & Log-loss & 0.408 & 0.063 \\
\bottomrule
\end{tabular}
\end{table}

\begin{figure}[htbp]
\centering
\begin{subfigure}[t]{0.48\textwidth}
\includegraphics[width=\textwidth]{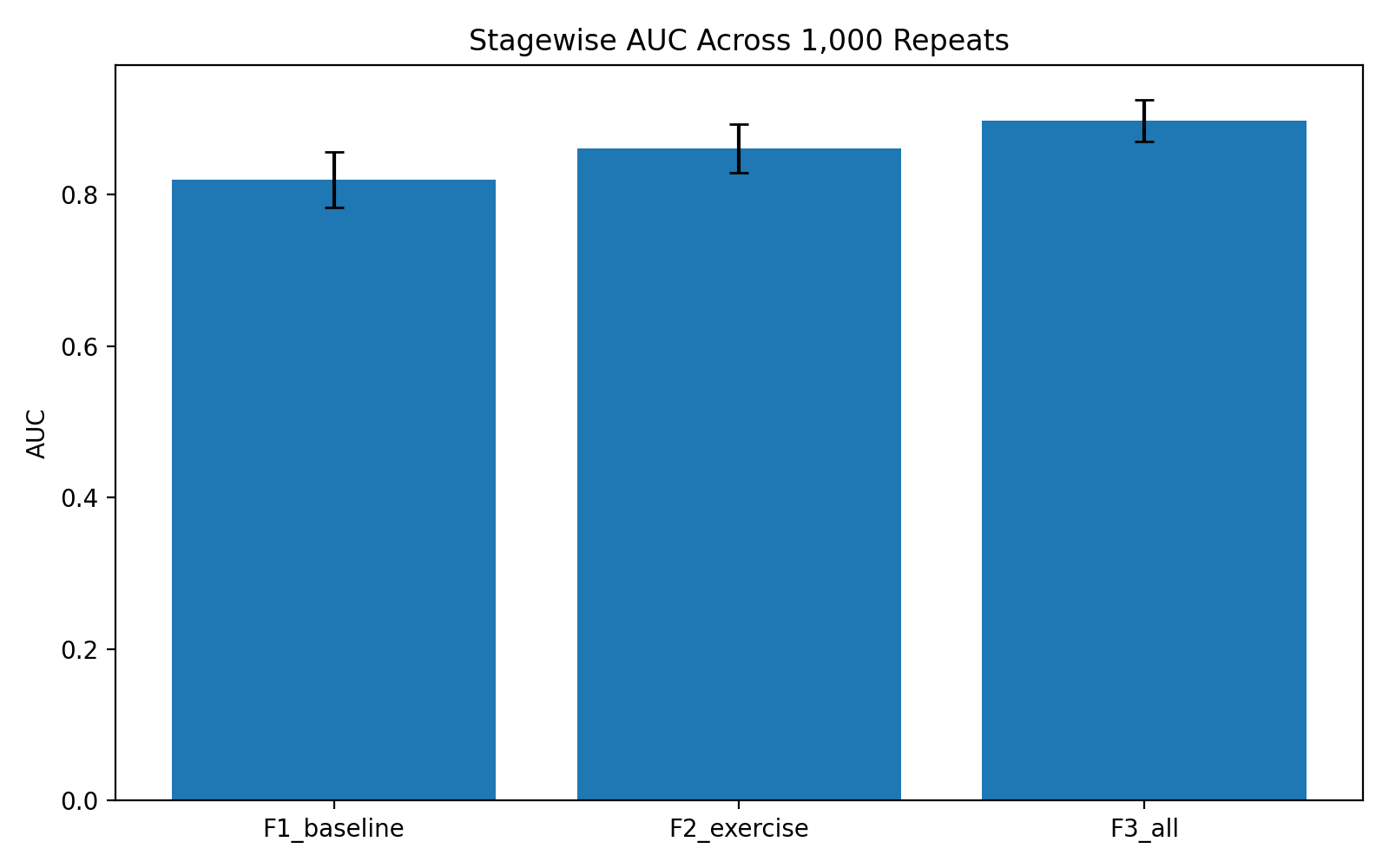}
\caption{Mean AUC at each stage.}
\end{subfigure}
\hfill
\begin{subfigure}[t]{0.48\textwidth}
\includegraphics[width=\textwidth]{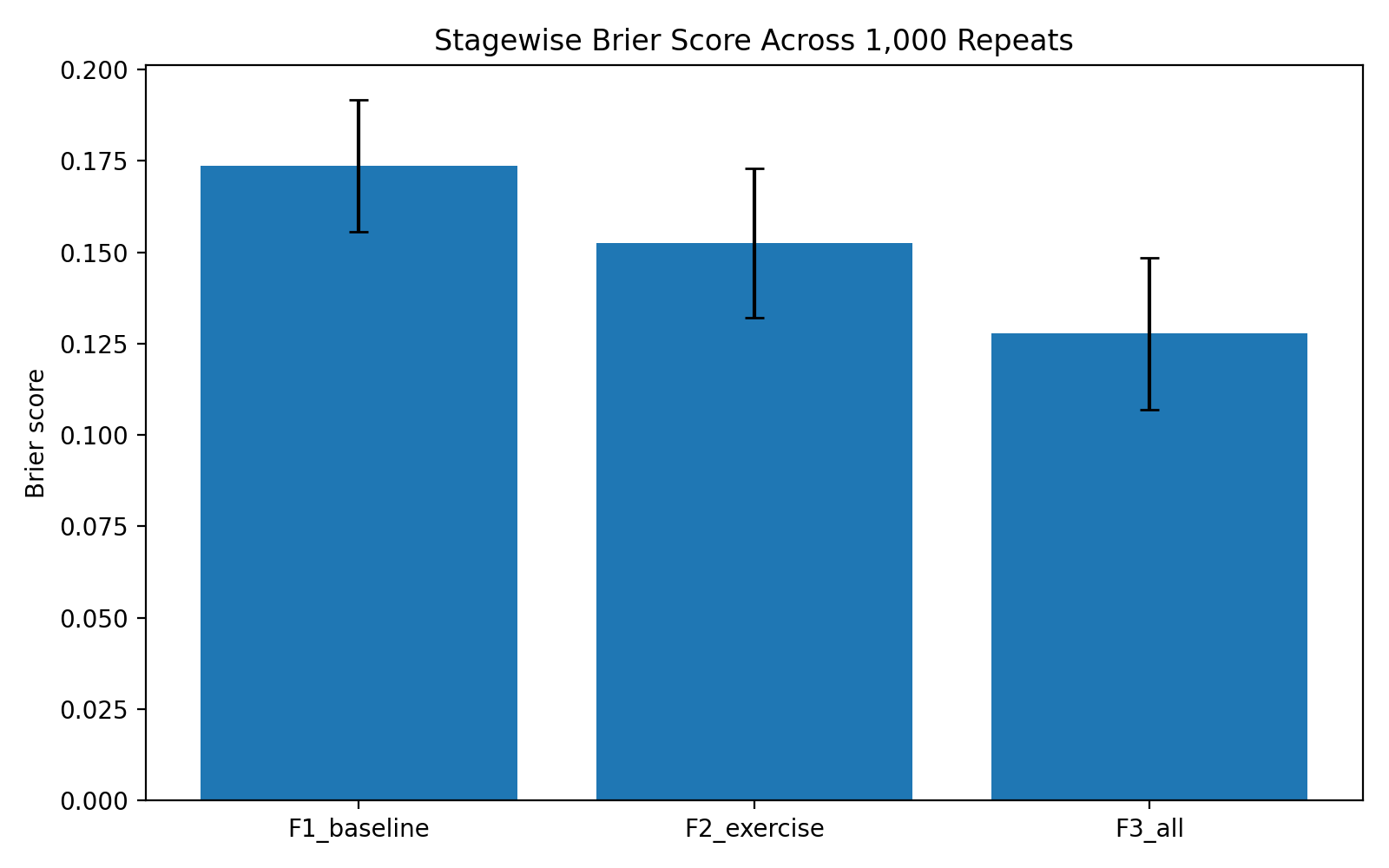}
\caption{Mean Brier score at each stage (lower is better).}
\end{subfigure}
\caption{Cleveland Heart Disease (Study~2): stagewise predictive performance
over 1{,}000 repeated experiments.  AUC and Brier both improve monotonically
from resting baseline to exercise testing and again after imaging/invasive
information.  Although performance improves at $\F_3$, the optimal-stopping
analysis in Table~\ref{tab:cleveland_decision} shows this improvement is
insufficient to justify the added stage cost.}
\label{fig:cleveland_performance}
\end{figure}

\subsubsection{Decision performance and optimal stopping}
\label{sec:cleveland_stop}

Table~\ref{tab:cleveland_decision} reports decision performance and stopping
costs.  The illustrative cumulative stage costs are $(0, 0.02, 0.06)$ for
$(\F_1, \F_2, \F_3)$; these summarize the added burden of the resting
assessment, exercise testing, and imaging/invasive workup respectively.

\begin{table}[htbp]
\centering
\caption{Decision performance and stopping analysis (Cleveland Heart Disease,
Study~2).  Illustrative test-cost schedule $(0, 0.02, 0.06)$.
$\dagger$~Minimum total cost.  The literal message of the table is that the
final-stage reduction in decision loss is too small to pay for the extra burden
of imaging/invasive workup, so the preferred stopping stage is $\F_2$.  The
table translates subject-specific risk updates into the average clinical
consequences of acting at each stage.}
\label{tab:cleveland_decision}
\renewcommand{\arraystretch}{1.15}
\begin{tabular}{lrrrrrr}
\toprule
Stage & Sensitivity & Specificity & Accuracy & Dec.\ loss & Test cost & Total \\
\midrule
$\F_1$ & 0.941 & 0.352 & 0.624 & 0.486 & 0.00 & 0.486 \\
$\F_2$ & 0.937 & 0.484 & 0.693 & 0.424 & 0.02 & 0.444$^\dagger$ \\
$\F_3$ & 0.917 & 0.586 & 0.738 & 0.416 & 0.06 & 0.476 \\
\bottomrule
\end{tabular}
\end{table}

The key finding is a direct contrast between predictive improvement and
clinical action.  Although AUC and Brier score improve from $\F_2$ to $\F_3$,
adding imaging/invasive information reduces decision loss by only $0.008$
(from $0.424$ to $0.416$), which is much smaller than the additional stage
cost of $0.04$.  Under this illustrative clinical trade-off, the preferred
stopping stage is therefore $\F_2$, with total cost $0.444$; moving to $\F_3$
increases total cost to $0.476$.

In practice, this suggests that routine escalation beyond the exercise-test
stage would not be justified for every patient when the resting and exercise
information already support a sufficiently informed decision.  The final stage
still predicts better, but the incremental gain may be too small to matter
clinically once the burden of later imaging is taken into account.  Read
literally, Table~\ref{tab:cleveland_decision} says that the final workup buys
only a small reduction in wrong decisions, and that small gain is not enough to
pay for the added burden encoded in the cost schedule.  This is exactly the
distinction the paper seeks to highlight: the best-looking population model is
not automatically the best stage for an individual patient's management.

\begin{figure}[htbp]
\centering
\includegraphics[width=0.58\textwidth]{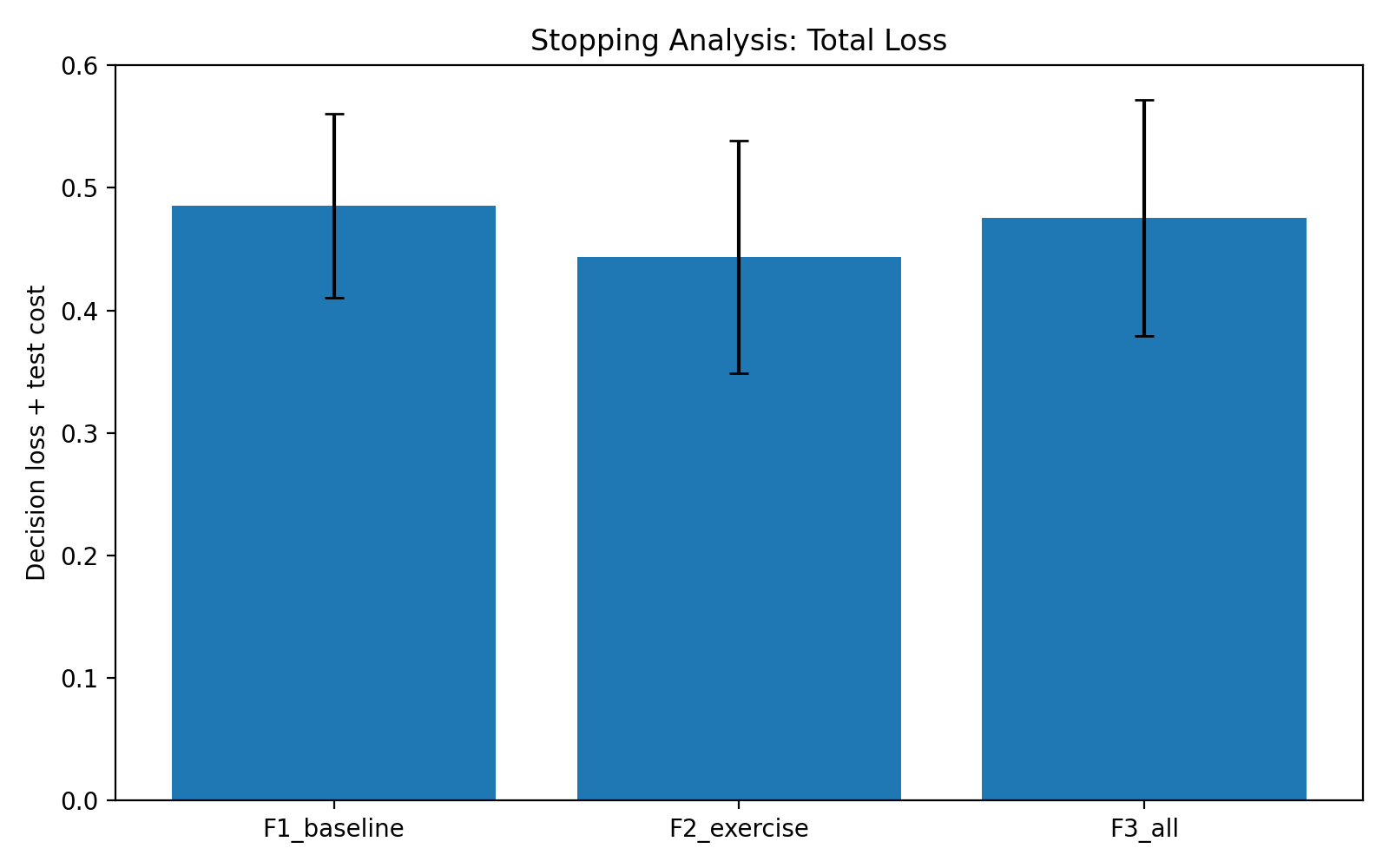}
\caption{Cleveland Heart Disease (Study~2): total expected cost (decision loss
+ cumulative test cost) by stage.  The exercise-test stage $\F_2$ minimizes
total cost despite $\F_3$ having lower decision loss in isolation.  This
illustrates how the optimal stopping stage need not be the stage with the best
predictive metrics.}
\label{fig:cleveland_stopping}
\end{figure}

\subsection{Study 3: Pima Diabetes --- When Early Action Is Preferred}
\label{sec:diabetes}

\subsubsection{Dataset and staged structure}
\label{sec:diabetes_data}

The Pima Indians Diabetes dataset \citep{smith1988} contains $n=768$ women of
Pima Indian heritage ($268$ diabetic, $500$ non-diabetic; 35\% prevalence),
each described by eight diagnostic measurements originally collected by the
National Institute of Diabetes and Digestive and Kidney Diseases.  We define
three information stages based on the clinical accessibility of the features:
\begin{align*}
\F_1 &= \{\text{glucose (fasting), BMI, age}\},\\
\F_2 &= \F_1 \cup \{\text{number of pregnancies, diabetes pedigree function}\},\\
\F_3 &= \text{all 8 features}~(\F_2 \cup \{\text{resting blood pressure,
        triceps skin thickness, serum insulin}\}).
\end{align*}
$\F_1$ captures the most accessible metabolic measurements; $\F_2$ adds reproductive and family history; $\F_3$ adds measurements requiring clinical examination or additional assay.

\subsubsection{Stagewise predictive performance}
\label{sec:diabetes_pred}

Table~\ref{tab:diabetes_stage} reports mean performance over 1{,}000 repeated
splits.  AUC improves from the screening stage $\F_1$ to the intermediate stage
$\F_2$ (from $0.827$ to $0.839$) but then declines slightly at the full stage
$\F_3$ (to $0.835$).  Brier score follows the same non-monotone pattern.
Unlike Studies~1 and~2, the most detailed stage does not improve prediction.
A plausible explanation is that skin-thickness and insulin measurements add
noise or instability in this dataset rather than clinically useful incremental
information.

The non-monotone AUC pattern suggests that the additional measurements add noise rather than useful signal in this dataset, making it a setting in which the optimal-stopping perspective directly clarifies why later measurement is not automatically better for action.

\begin{table}[htbp]
\centering
\caption{Stagewise predictive performance over 1{,}000 repeated train/test
splits (Pima Diabetes, Study~3).  Mean and SD across repetitions.  The table
summarizes cohort-level evidence about subject-specific risk updating under a
staged screening pathway.}
\label{tab:diabetes_stage}
\renewcommand{\arraystretch}{1.15}
\begin{tabular}{llrr}
\toprule
Stage & Metric & Mean & SD \\
\midrule
$\F_1$ (screen)       & AUC      & 0.827 & 0.023 \\
                       & Brier    & 0.161 & 0.011 \\
                       & Accuracy & 0.765 & 0.022 \\
                       & Log-loss & 0.486 & 0.027 \\
\addlinespace
$\F_2$ ($\F_1{}+{}$pregnancies/pedigree) & AUC & 0.839 & 0.022 \\
                       & Brier    & 0.156 & 0.011 \\
                       & Accuracy & 0.769 & 0.022 \\
                       & Log-loss & 0.476 & 0.030 \\
\addlinespace
$\F_3$ (all 8 feat.)  & AUC      & 0.835 & 0.023 \\
                       & Brier    & 0.158 & 0.011 \\
                       & Accuracy & 0.767 & 0.022 \\
                       & Log-loss & 0.480 & 0.030 \\
\bottomrule
\end{tabular}
\end{table}

\begin{figure}[htbp]
\centering
\begin{subfigure}[t]{0.48\textwidth}
\includegraphics[width=\textwidth]{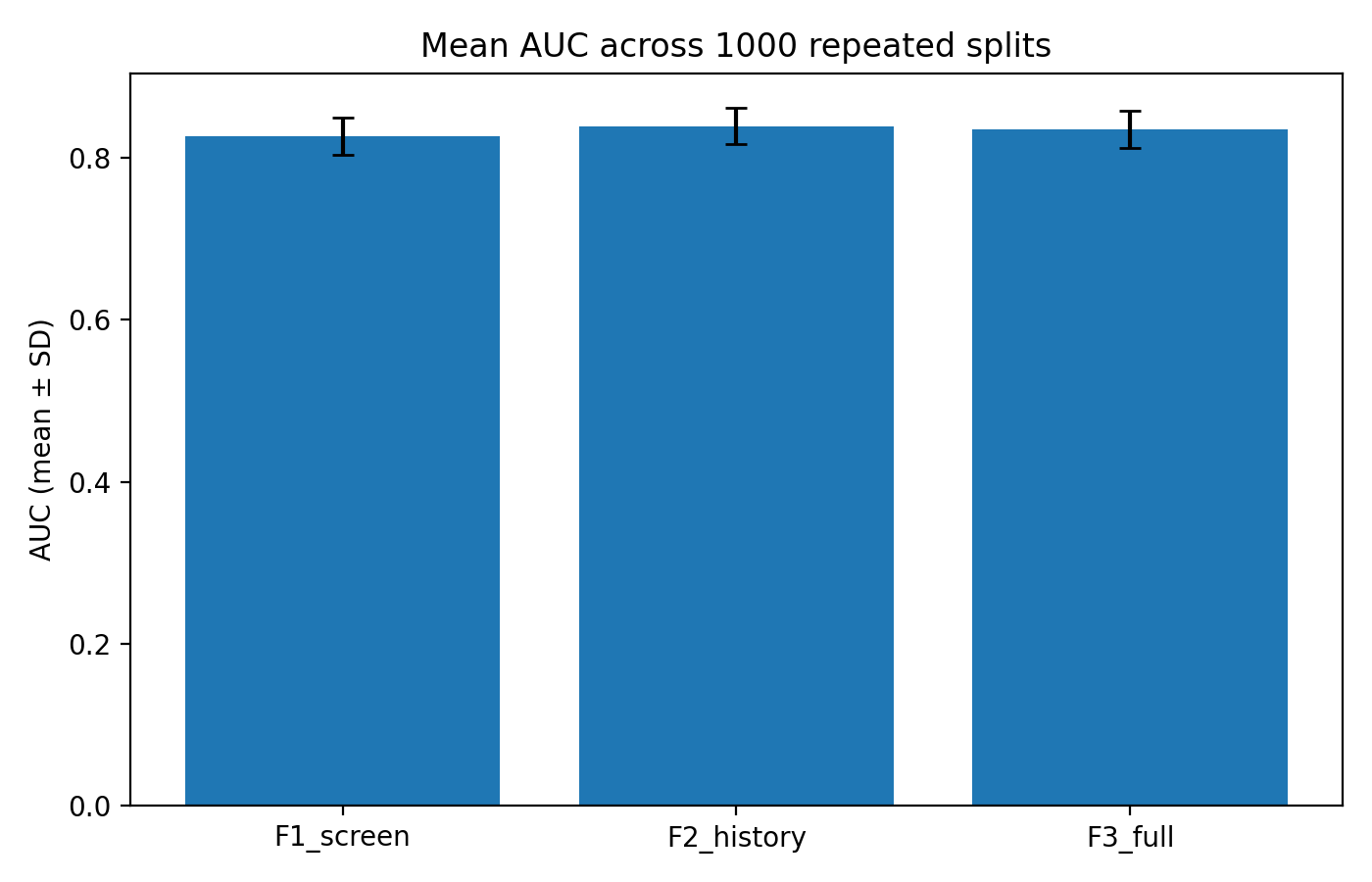}
\caption{Mean AUC at each stage.}
\end{subfigure}
\hfill
\begin{subfigure}[t]{0.48\textwidth}
\includegraphics[width=\textwidth]{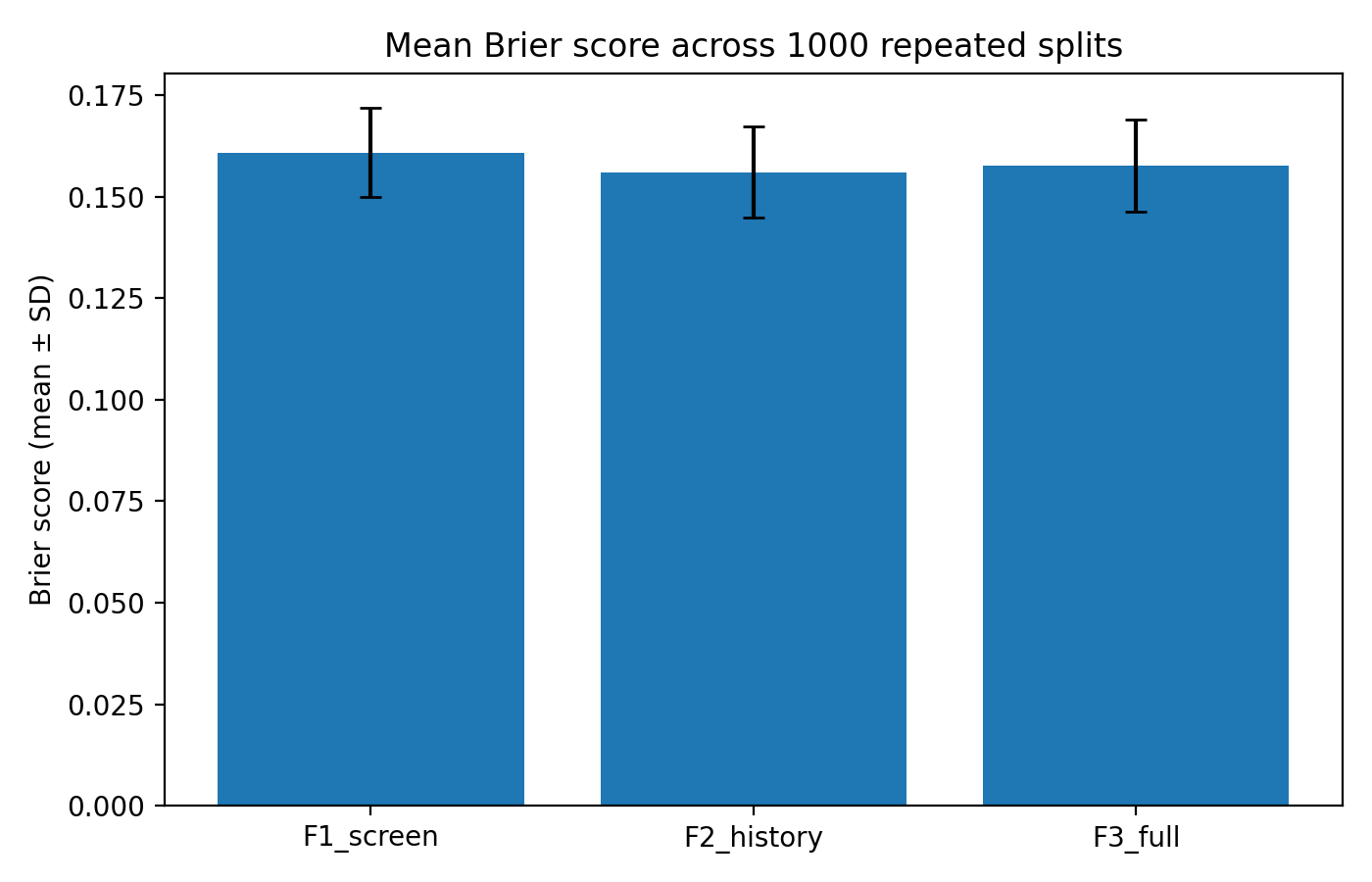}
\caption{Mean Brier score at each stage (lower is better).}
\end{subfigure}
\caption{Pima Diabetes (Study~3): stagewise predictive performance over
1{,}000 repeated experiments.  Unlike Studies~1 and~2, discrimination is
non-monotone: AUC peaks at $\F_2$ and declines slightly at the full-feature
stage $\F_3$, suggesting that the additional measurements add noise rather
than clinically useful signal.}
\label{fig:diabetes_performance}
\end{figure}

\subsubsection{Decision performance and optimal stopping}
\label{sec:diabetes_stop}

Table~\ref{tab:diabetes_decision} reports decision performance and total cost
under the cost ratio $(c_{\mathrm{FP}}, c_{\mathrm{FN}}) = (1, 5)$ and the
illustrative cumulative stage costs $(0, 0.01, 0.03)$.  As throughout the
paper, these stage costs are intended to summarize the additional burden of
further assessment rather than literal fees.

\begin{table}[htbp]
\centering
\caption{Decision performance and stopping analysis (Pima Diabetes, Study~3).
Test-cost schedule $(0, 0.01, 0.03)$.  $\dagger$~Minimum total cost.  The
literal message is that the modest gain at $\F_2$ is not enough to justify even
a small extra stage cost, and the full stage is worse still, so the preferred
action is to stop at $\F_1$.  The table converts stagewise risk updating into
the average consequences of early versus delayed action.}
\label{tab:diabetes_decision}
\renewcommand{\arraystretch}{1.15}
\begin{tabular}{lrrrrrr}
\toprule
Stage & Sensitivity & Specificity & Accuracy & Dec.\ loss & Test cost & Total \\
\midrule
$\F_1$ & 0.942 & 0.469 & 0.635 & 0.446 & 0.00 & 0.446$^\dagger$ \\
$\F_2$ & 0.934 & 0.504 & 0.655 & 0.437 & 0.01 & 0.447 \\
$\F_3$ & 0.929 & 0.502 & 0.652 & 0.448 & 0.03 & 0.478 \\
\bottomrule
\end{tabular}
\end{table}

The preferred stopping stage is $\F_1$: acting on the three basic screening
measurements (glucose, BMI, age) minimizes total expected cost under this
illustrative trade-off.
Although adding pregnancies and pedigree function at $\F_2$ reduces decision
loss from $0.446$ to $0.437$, the gain of $0.009$ is smaller than the added
stage cost of $0.01$.  At $\F_3$, decision loss rises back above the $\F_1$ level
(to $0.448$), so the full stage is not attractive even before additional stage
costs are considered.

Clinically, this suggests that the earliest screening information may already
be sufficient for the immediate action represented by the chosen loss function.
In this dataset, ordering the later measurements does not appear to change
management enough to justify waiting for them.  Read literally,
Table~\ref{tab:diabetes_decision} says that the clinician would usually make
the same action choice after the later measurements as after the initial
screen, so the extra measurements add little decision value under this trade-off.
The point is not that the cohort-level model at $\F_3$ is uninteresting; it is
that the individual patient's risk estimate often does not move enough to alter
the decision.

\begin{figure}[htbp]
\centering
\includegraphics[width=0.58\textwidth]{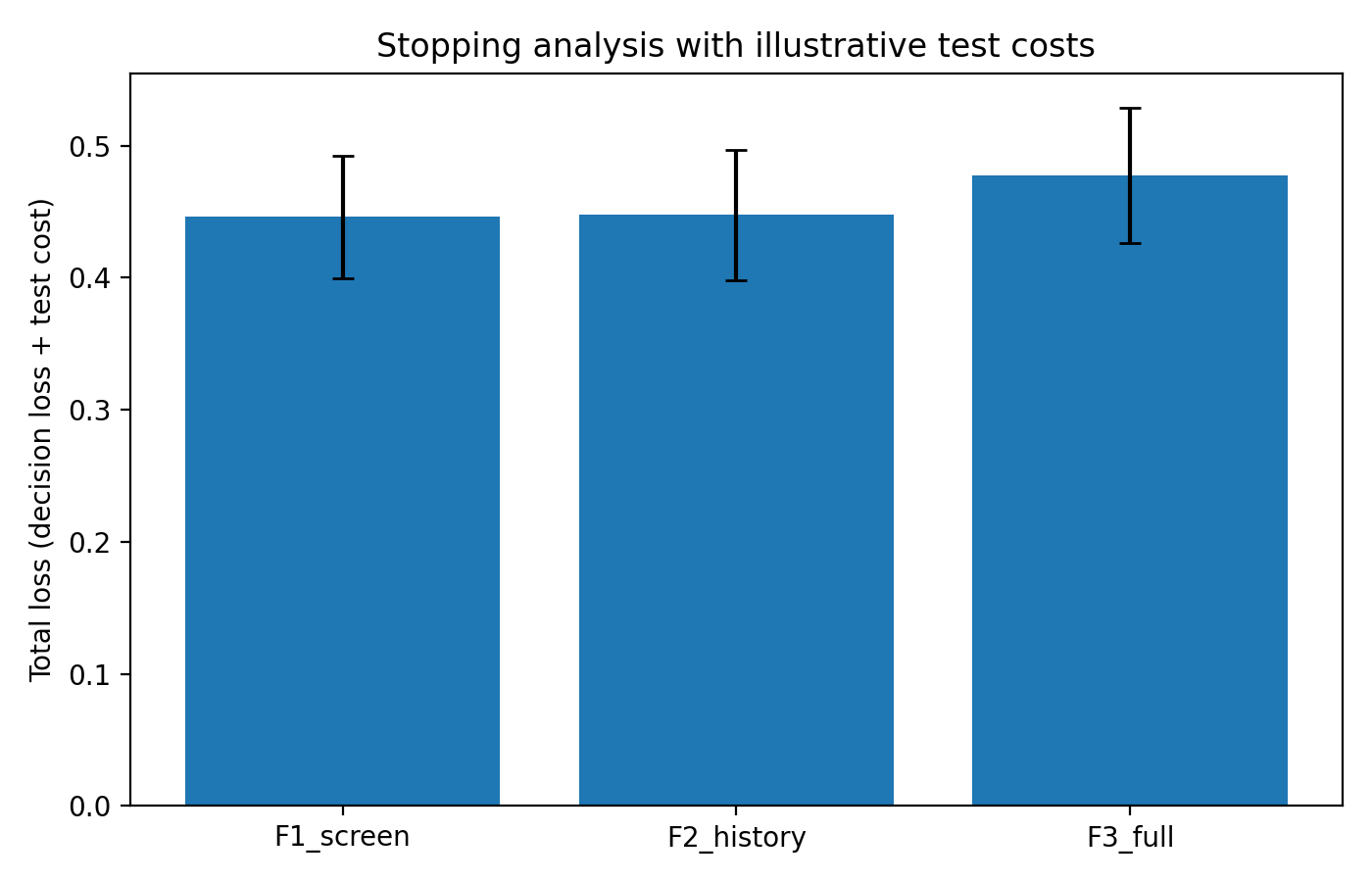}
\caption{Pima Diabetes (Study~3): total expected cost (decision loss +
cumulative test cost) by stage.  The screening stage $\F_1$ minimizes total
cost.  The gain from adding pregnancies and pedigree at $\F_2$ is smaller than
the added stage cost, and the full-feature stage $\F_3$ has higher decision
loss than $\F_1$ before any stage cost is added.}
\label{fig:diabetes_stopping}
\end{figure}

\subsection{Study 4: eICU Collaborative Research Database Demo --- Sequential Prediction of In-Hospital Mortality}
\label{sec:eicu}

\subsubsection{Dataset and staged structure}
\label{sec:eicu_data}

The eICU Collaborative Research Database Demo \citep{pollard2018eicu} provides
$n=2{,}359$ ICU admissions (after requiring at least one 24-hour vital-sign
observation) drawn from the 2014--2015 eICU-CRD and freely available from
PhysioNet without credentialed access.  In-hospital mortality is 8.6\%.  The
three information stages $\F_1 \subset \F_2 \subset \F_3$ are as follows:
$\F_1$ aggregates eight vital-sign means over the first 24 ICU hours; $\F_2$
adds sixteen first-available laboratory values; and $\F_3$ adds five static
demographic/admission variables (age, sex, ICU unit type, admission height and
weight).  This ordering is analytic (by incremental predictive value) rather
than strictly chronological; an alternative ordering is examined in
Appendix Section~\ref{sec:S4}.  We use $c_{\mathrm{FP}}=1$,
$c_{\mathrm{FN}}=10$ ($c^*=1/11\approx 0.091$), and the same illustrative
cumulative stage costs $(0,0.01,0.03)$ as in Studies~1 and~3.

\begin{table}[htbp]
\centering
\caption{Stagewise predictive performance over 1{,}000 repeated train/test
splits (eICU Demo, Study~4).  Mean and SD across repetitions.
$c^* \approx 0.091$ ($c_{\mathrm{FN}}=10$).}
\label{tab:eicu_stage}
\renewcommand{\arraystretch}{1.15}
\begin{tabular}{llrr}
\toprule
Stage & Metric & Mean & SD \\
\midrule
$\F_1$ (vitals, 8 feat.)        & AUC         & 0.699 & 0.032 \\
                                 & Brier       & 0.209 & 0.006 \\
                                 & Sensitivity & 0.989 & 0.012 \\
                                 & Specificity & 0.006 & 0.004 \\
                                 & Accuracy    & 0.091 & 0.004 \\
\addlinespace
$\F_2$ (vitals + labs, 24 feat.) & AUC         & 0.731 & 0.030 \\
                                 & Brier       & 0.190 & 0.007 \\
                                 & Sensitivity & 0.982 & 0.018 \\
                                 & Specificity & 0.021 & 0.011 \\
                                 & Accuracy    & 0.103 & 0.009 \\
\addlinespace
$\F_3$ (full, 29 feat.)          & AUC         & 0.770 & 0.027 \\
                                 & Brier       & 0.180 & 0.008 \\
                                 & Sensitivity & 0.980 & 0.017 \\
                                 & Specificity & 0.126 & 0.024 \\
                                 & Accuracy    & 0.200 & 0.021 \\
\bottomrule
\end{tabular}
\end{table}

\subsubsection{Stagewise predictive performance}
\label{sec:eicu_pred}

Table~\ref{tab:eicu_stage} reports mean predictive performance over $1{,}000$
repeated 70/30 train/test splits.  AUC improved monotonically from $0.699$ at
$\F_1$ to $0.731$ at $\F_2$ and $0.770$ at $\F_3$.  Brier score fell
correspondingly from $0.209$ to $0.190$ and $0.180$.  These improvements are
consistent and statistically stable (SDs of 0.027--0.032 on AUC).

\begin{figure}[htbp]
\centering
\begin{subfigure}[t]{0.88\textwidth}
\includegraphics[width=\textwidth]{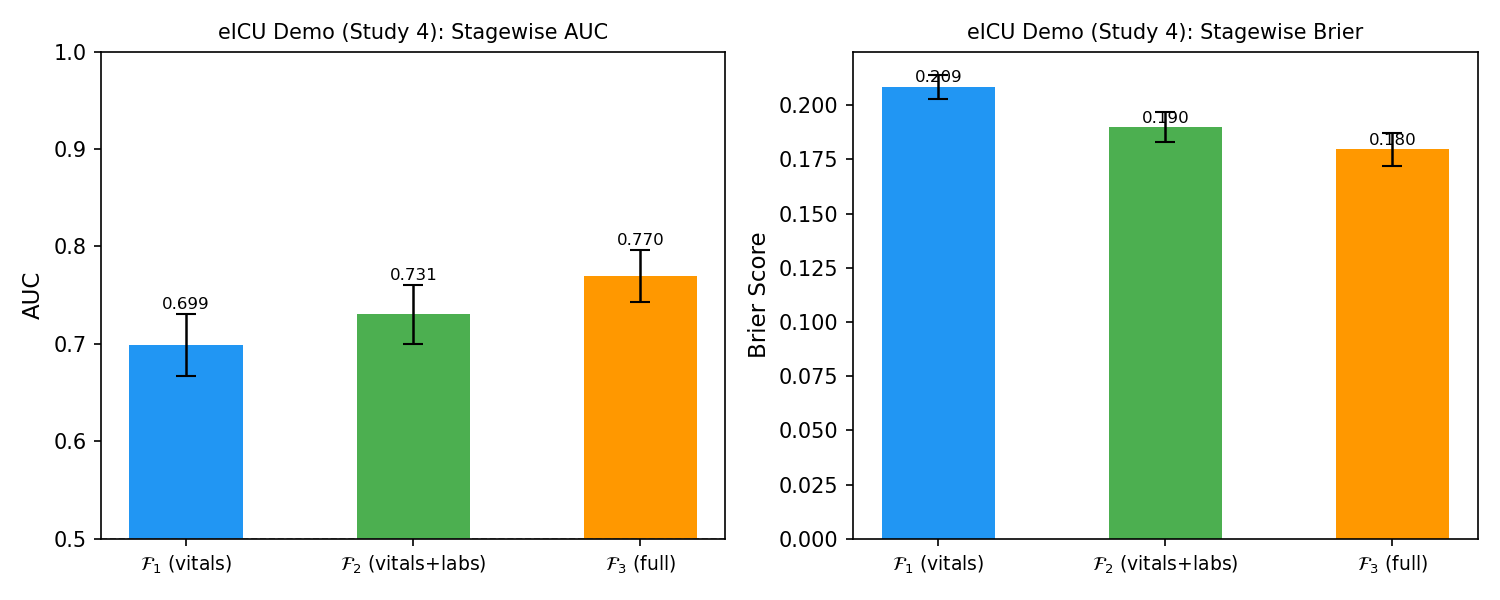}
\caption{Mean AUC and Brier at each stage (mean $\pm$ SD, 1{,}000 reps).}
\end{subfigure}
\caption{eICU Demo (Study~4): stagewise predictive performance.  AUC improves
monotonically and Brier score decreases monotonically as clinical information
accumulates.  Error bars represent $\pm$1 SD across 1{,}000 repeated splits.}
\label{fig:eicu_performance}
\end{figure}

Sensitivity is near-perfect at all stages but specificity is essentially zero at $\F_1$ and $\F_2$, rising meaningfully only at $\F_3$ (0.126), where demographic context enables identification of a low-risk subset.

\begin{table}[htbp]
\centering
\caption{Decision performance and stopping analysis (eICU Demo, Study~4).
Illustrative test-cost schedule $(0, 0.01, 0.03)$.  $\dagger$~Minimum total
cost.  The large decision losses at $\F_1$ and $\F_2$ reflect the near-zero
specificity: under the low threshold $c^*\approx 0.091$, almost every patient
is flagged at early stages, generating a high false-positive burden.  Only
at $\F_3$ does the demographic information enable a meaningful reduction in
unnecessary escalation.}
\label{tab:eicu_decision}
\renewcommand{\arraystretch}{1.15}
\begin{tabular}{lrrrrrr}
\toprule
Stage & Sensitivity & Specificity & Accuracy & Dec.\ loss & Test cost & Total \\
\midrule
$\F_1$ & 0.989 & 0.006 & 0.091 & 0.918 & 0.00 & 0.918 \\
$\F_2$ & 0.982 & 0.021 & 0.103 & 0.911 & 0.01 & 0.921 \\
$\F_3$ & 0.980 & 0.126 & 0.200 & 0.816 & 0.03 & 0.846$^\dagger$ \\
\bottomrule
\end{tabular}
\end{table}

\subsubsection{Decision performance and optimal stopping}
\label{sec:eicu_stopping}

Table~\ref{tab:eicu_decision} reports decision performance and total cost under
the stated parameters.  Expected decision loss falls monotonically from $0.918$
at $\F_1$ to $0.911$ at $\F_2$ and $0.816$ at $\F_3$.  Adding the cumulative
stage costs $(0,0.01,0.03)$ gives total costs $(0.918, 0.921, 0.846)$, so the
full-information stage $\F_3$ minimizes total cost.

The large reduction in decision loss from $\F_2$ to $\F_3$ ($0.911 \to
0.816$, a drop of $0.095$) far exceeds the stage cost of $0.02$, making
$\F_3$ the clear optimal stopping stage.  In contrast, moving from $\F_1$ to
$\F_2$ reduces decision loss by only $0.007$ while adding a stage cost of
$0.01$, making that transition marginally harmful from a cost-benefit
perspective.  In practical terms: for in-hospital mortality prediction under
this cost schedule, routine labs add little beyond vital signs alone, but
demographic context substantially reduces the rate of unnecessary escalation.

\begin{figure}[htbp]
\centering
\includegraphics[width=0.58\textwidth]{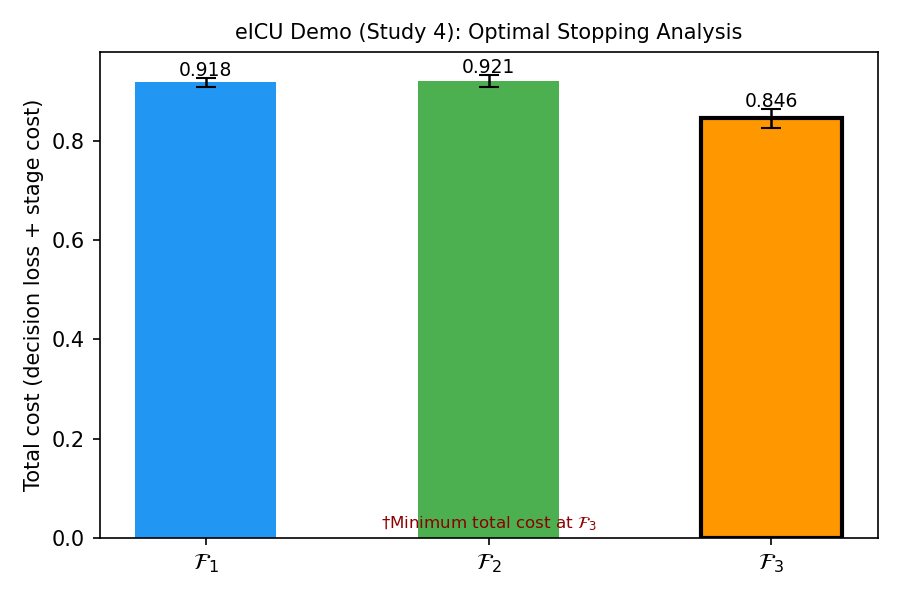}
\caption{eICU Demo (Study~4): total expected cost (decision loss + cumulative
test cost) by stage.  The full-information stage $\F_3$ minimizes total cost,
with the large drop in decision loss from $\F_2$ to $\F_3$ far outweighing
the stage cost.}
\label{fig:eicu_stopping}
\end{figure}

\section{Discussion}
\label{sec:discussion}

\subsection{What the optimal-stopping perspective adds to clinical prediction}
\label{sec:disc_framework}

Most clinical prediction papers report performance at a fixed feature set and
stop there.  The present framework is intended to support a different and more
clinically direct question: \emph{when should testing stop and action begin?}
Across the four applications, the answer differed materially by setting.  In
Study~1 (Breast Cancer Wisconsin), fuller information improved prediction
enough to justify the added stage cost.  In Study~2 (Cleveland Heart Disease),
later-stage testing improved prediction but not enough to justify routine
escalation beyond the exercise-test stage under the assumed trade-off.  In
Study~3 (Pima Diabetes), the earliest screening stage already supported the
lowest total cost.  In Study~4 (eICU), the full-information stage was clearly
preferred, with the final transition producing the largest decision-loss
reduction across all four studies.

These contrasts matter because they show that ``best predictive model'' and
``best stage to act'' are not the same claim.  From a clinical perspective, the
important output is not only whether performance improves, but whether the
improvement is large enough to change management in a worthwhile way.  That is
why optimal stopping is a useful highlight of the paper: it turns sequential
prediction from a descriptive exercise into a decision-analytic one.

Study~4 (eICU) adds a further dimension.  Here, the full-information stage is
clearly optimal, but for an instructive reason: early-stage models based on
vital signs and laboratories provide near-perfect sensitivity but essentially
no specificity under the low clinical threshold ($c^* \approx 0.091$).  The
decision loss at early stages is dominated by the cost of flagging every
patient as high risk.  It is the addition of demographic features that first
enables a meaningful fraction of patients to be correctly classified as low
risk, producing the largest drop in decision loss across all four studies
($0.095$ from $\F_2$ to $\F_3$).  This illustrates that the value of an
additional stage can be non-uniform: a transition that adds modest AUC
improvement may nonetheless produce a large decision gain if it specifically
improves threshold-relevant stratification.

A second point that deserves emphasis is the level at which the framework
operates.  Standard classification studies are usually population based and
cross-sectional: they estimate a mapping from features to outcome using one
fixed feature set, then summarize performance over the cohort.  Our framework
still relies on such data for estimation, but its interpretation is
subject-specific.  The object of interest is the evolving conditional risk for
an individual patient as additional information becomes available.  In that
sense, the martingale structure is not a decorative reformulation of ``more
variables improve prediction''; it encodes coherence of the patient's risk
trajectory across nested information sets.  By contrast, the reverse-martingale
idea is reserved for the different problem of projecting a richer predictor onto
a simpler clinical score or an earlier-stage summary.  The connections to
decision curve analysis and to the treatment-selection extension of the framework
are discussed in the appendix.

\subsection{Approximate vs.\ exact martingale behavior}
\label{sec:disc_approx}

The martingale property should be viewed here as an interpretive benchmark
rather than as a literal claim that fitted risks are exact conditional
expectations.  In practice, stagewise models are estimated from finite samples,
may be imperfectly calibrated, and may use variables that are only approximately
nested in the order observed in clinical care.  The drift diagnostic is useful
because it checks whether stage-to-stage updates behave broadly as expected
under a coherent information sequence.  In Studies~1--3, the drift values were small, suggesting that the staged
formulation is a reasonable description of the data even when later-stage
prediction is not monotone.  Study~4 exhibited systematic negative drift
($M\approx -0.03$ to $-0.04$ at both transitions), a striking contrast with
the near-zero drift in the first three studies.  This contrast is itself
clinically informative: it reflects a fundamental difference in the nature
of the staged information.  In Studies~1--3, later features refine a risk
estimate that is already in a reasonable range, so corrections are small and
symmetric.  In Study~4, the early-stage vital-sign model lacks the baseline
health information needed to stratify the low-risk tail, so it assigns nearly
everyone a predicted probability above the threshold; adding demographic
context systematically revises estimates downward for lower-risk patients.
The martingale benchmark is therefore useful not only as a coherence check
but as a diagnostic of \emph{what kind of information} each stage contributes:
approximate martingale behavior (Studies~1--3) indicates refinement of an
adequate early estimate, while systematic directional drift (Study~4)
indicates that early models are structurally limited and a later feature
block is doing qualitatively different work.  The drift diagnostic in Study~4
therefore serves to explain \emph{why} the final stage produces such a large
reduction in decision loss, rather than to flag incoherence.  The same caution applies to the reverse-martingale
language: it should be read as a statement about principled information
coarsening, not as a claim that forward clinical improvement itself runs in
reverse time.

\begin{remark}
\textbf{Forward and reverse martingales: a clinical summary.}
The individual patient's workup is a \emph{forward-information} problem: posterior
risk is updated under an increasing filtration, and optimal stopping determines
when one more measurement is no longer worth its burden.  Reverse-martingale
ideas apply separately, to \emph{deliberate compression} of rich information into
a simpler score; the reverse-martingale projection loss quantifies what is lost
by that simplification.
\end{remark}

\subsection{Limitations}
\label{sec:disc_limits}

\begin{enumerate}[leftmargin=1.8em,label=(\roman*),itemsep=2pt]
\item All four applications are retrospective, and stage definitions are
  analyst-constructed approximations to clinical workup pathways.  The stopping
  recommendations should be read as methodological illustrations rather than
  practice-ready decision rules.

\item The stage-cost values are illustrative summaries of burden, not formal
  health-economic estimates, and may not match the preferences of specific
  clinicians, patients, or healthcare systems.

\item The datasets are relatively small and, in some cases, old benchmark
  collections, and the framework is presently limited to binary outcomes and
  threshold-based decisions.  Extensions to time-to-event outcomes and richer
  treatment choices are important directions for future work.

\item The stopping recommendations are derived from cohort-average expected
  losses rather than true patient-specific Bellman values.  The patient-level
  bridge quantities in Table~\ref{tab:patient_level} provide a partial remedy,
  but fully individualized implementation would require a prospective design or
  a carefully validated counterfactual model.

\item In Study~4 (eICU), demographics placed at $\F_3$ are typically available
  at admission; the staging is therefore analytic rather than temporal.  A
  sensitivity ordering with demographics at $\F_1$ is presented in
  Appendix Section~\ref{sec:S4} and does not change the
  qualitative findings.
\end{enumerate}

\section{Conclusion}
\label{sec:conclusion}

This paper argues that \emph{optimal stopping} is a useful way to frame
sequential clinical prediction.  The central applied question is not only how
risk estimates change as more information becomes available, but when enough
information has been collected to justify action.  Across four retrospective
studies, the preferred stage for action differed meaningfully by clinical
setting: waiting for fuller information was worthwhile in Breast Cancer
Wisconsin and eICU in-hospital mortality, an intermediate stage was preferred
in Cleveland Heart Disease, and the earliest screening stage was preferred in
Pima Diabetes under the illustrative loss and cost assumptions.  Patient-level
stability analysis (Table~\ref{tab:patient_level}) provided a direct
interpretation: in Pima Diabetes, 98\% of patients already had a stable
decision recommendation at $\F_1$, explaining why further testing added
negligible value; in eICU, 11\% of patients changed their recommended action
at the $\F_2\to\F_3$ transition, generating the large observed reduction in
decision loss.

The eICU study further illustrates that AUC improvement does not automatically
translate into decision improvement.  Vital signs and lab values alone were
insufficient to identify low-risk ICU patients under a mortality-sensitive cost
schedule; demographic context was needed to produce meaningful specificity.
The optimal-stopping framework captured this nuance directly: the large
decision loss reduction at $\F_3$ justified the added stage cost despite modest
AUC gains.

That message is stronger than the familiar observation that more variables can
improve prediction.  The models are estimated from retrospective population
data, but the scientific interpretation is subject-specific.  The tables and
figures should be read as cohort-level summaries of an underlying patient-level
decision process: how the conditional risk for a given patient changes as the
workup unfolds, whether that change is coherent across stages, and whether it
is large enough to justify one more test before acting.  Forward martingale
structure captures that staged updating; reverse-martingale structure pertains
to deliberate coarsening of richer information into simpler summaries.

Emphasizing optimal stopping also clarifies why the empirical results are not
trivial.  The contribution is not merely that later models may fit better on
average.  It is that retrospective prediction models can be linked into an
interpretable subject-specific trajectory that supports a clinically important
decision: when to keep testing and when to stop.

\section*{Acknowledgements}

The author thanks the UCI Machine Learning Repository for hosting the Breast
Cancer Wisconsin, Cleveland Heart Disease, and Pima Diabetes datasets, and
PhysioNet for hosting the eICU Collaborative Research Database Demo.  This
work is part of the NSC 2026 RMRNN project.

\section*{Data availability}

All datasets are freely and publicly accessible.  The Breast Cancer Wisconsin,
Cleveland Heart Disease, and Pima Diabetes datasets are from the UCI ML
Repository (\url{https://archive.ics.uci.edu/}).  The eICU Demo is from
PhysioNet (\url{https://physionet.org/content/eicu-crd-demo/}) under a
Creative Commons license.  The analysis code is described in
Appendix Section~\ref{sec:S3}.

\section*{Conflict of interest}

The author declares no conflict of interest.

\section*{Appendix Overview}

The appendices are included in this manuscript.  Appendix~\ref{app:proofs}
contains the proofs of the main theoretical results.  The S-numbered appendix
sections contain the material previously prepared as supplementary material:
Section~\ref{sec:S1} gives the proof of Proposition~\ref{prop:regret}
(decision regret under posterior compression), Section~\ref{sec:S2} reports
the projection-loss analyses and martingale drift diagnostics for all four
studies, Section~\ref{sec:S3} lists the analysis scripts and software
dependencies, Section~\ref{sec:S4} presents stage-cost sensitivity analyses
for Studies~2 and~4 and an alternative staging order for the eICU study, and
Section~\ref{sec:S5} reports calibration slopes and intercepts for each
dataset and stage.
\bibliographystyle{unsrtnat}


\appendix

\section{Proofs of Main Results}
\label{app:proofs}

Throughout, $(\Omega,\mathcal{F},\Pbb)$ is a fixed probability space.
All conditional expectations are versions of the Radon--Nikod\'{y}m
derivative and are identified up to a.s.\ equivalence.
We write $L^p$ for $L^p(\Omega,\mathcal{F},\Pbb)$.
The standard references for martingale theory used here are
\citet{doob1953} and \citet{williams1991}; for optimal stopping,
\citet{chow1971os} and \citet{peskir2006}; for Bayes decision theory,
\citet{berger1985}.

\subsection{Proof of Proposition~\ref{prop:martingale}
(Posterior risk is a martingale)}
\label{proof:martingale}

\textbf{Integrability.}
Since $D\in\{0,1\}\subset L^\infty\subseteq L^1$, the conditional expectation
$X_t=\E[D\mid\F_t]$ is well-defined for each $t\ge 0$, with
$\E|X_t|\le\E|D|<\infty$.

\textbf{Martingale property.}
For $s<t$, $\F_s\subseteq\F_t$, so the tower property of conditional
expectation \citep[Ch.~9]{williams1991} gives
\[
   \E[X_t\mid\F_s]
   \;=\; \E\bigl[\E[D\mid\F_t]\bigm|\F_s\bigr]
   \;=\; \E[D\mid\F_s]
   \;=\; X_s \quad\text{a.s.}
\]

\textbf{Boundedness.}
Because $D\in\{0,1\}$, we have $X_t=\Pbb(D=1\mid\F_t)\in[0,1]$ a.s.,
so $(X_t)$ is bounded in $L^\infty$ and in particular uniformly integrable.

\textbf{Convergence.}
A uniformly integrable martingale converges a.s.\ and in $L^1$ to an
integrable limit \citep[Thm.~11.5]{williams1991} (Doob's $L^1$ convergence
theorem).  The limit is $X_\infty=\E[D\mid\F_\infty]$
a.s., where $\F_\infty=\bigvee_{t\ge 0}\F_t$, by L\'{e}vy's upward theorem
\citep[Thm.~14.2]{williams1991}. \qed

\subsection{Proof of Proposition~\ref{prop:revmart}
(Projected risk is a reverse martingale)}
\label{proof:revmart}

\textbf{Integrability.}
$Y_t=\E[D\mid\G_t]$ satisfies $\E|Y_t|\le\E|D|<\infty$ since $D\in L^1$.

\textbf{Reverse martingale property.}
Since $(\G_t)$ is decreasing ($\G_{t+1}\subseteq\G_t$), the tower property
gives
\[
   \E[Y_t\mid\G_{t+1}]
   \;=\; \E\bigl[\E[D\mid\G_t]\bigm|\G_{t+1}\bigr]
   \;=\; \E[D\mid\G_{t+1}]
   \;=\; Y_{t+1} \quad\text{a.s.,}
\]
which is the defining reverse-martingale identity
\citep[Sec.~14]{williams1991}.

\textbf{Boundedness and convergence.}
Because $D\in[0,1]$, the tower property gives $Y_t=\E[D\mid\G_t]\in[0,1]$
a.s., so the family $\{Y_t\}$ is bounded in $L^\infty$ and hence uniformly
integrable.  Doob's reverse martingale convergence theorem
\citep[Thm.~14.4]{williams1991} then guarantees $Y_t\to Y_\infty$ a.s.\ and
in $L^1$, where $Y_\infty=\E[D\mid\G_\infty]$ and
$\G_\infty=\bigcap_{t\ge 0}\G_t$. \qed

\subsection{Proof of Proposition~\ref{prop:projection_loss}
(Projection-loss decomposition)}
\label{proof:projection}

Assume $D\in L^2$; then $X_t,Y_t\in L^2$ as well
(Jensen: $\E[X_t^2]\le\E[D^2]<\infty$).

Since $\G_t\subseteq\F_t$, every $\G_t$-measurable function is
$\F_t$-measurable, so $Y_t$ is $\F_t$-measurable.  Write the decomposition
\[
   D-Y_t \;=\; (D-X_t)+(X_t-Y_t).
\]
The cross-term vanishes: because $D\in L^2$ and $X_t=\E[D\mid\F_t]$ is the
$L^2$-orthogonal projection of $D$ onto $L^2(\F_t)$
\citep[Ch.~9]{williams1991}, the residual $D-X_t$ is orthogonal to every
$\F_t$-measurable $L^2$ function.  In particular $X_t-Y_t\in L^2(\F_t)$,
so
\[
   \E\bigl[(D-X_t)(X_t-Y_t)\bigr]
   = \E\bigl[\E\bigl[(D-X_t)(X_t-Y_t)\bigm|\F_t\bigr]\bigr]
   = \E\bigl[(X_t-Y_t)\,\underbrace{\E[D-X_t\mid\F_t]}_{=\,0}\bigr]
   = 0.
\]
Expanding $\E[(D-Y_t)^2]$ and using this orthogonality yields
\[
   \E[(D-Y_t)^2]
   = \E[(D-X_t)^2]+2\,\underbrace{\E[(D-X_t)(X_t-Y_t)]}_{=\,0}
     +\E[(X_t-Y_t)^2]
   = \E[(D-X_t)^2]+\E[(X_t-Y_t)^2].
\]
Since $\E[(X_t-Y_t)^2]\ge 0$, the inequality $\E[(D-Y_t)^2]\ge\E[(D-X_t)^2]$
follows.  Rearranging gives the stated identity. \qed

\subsection{Proof of Proposition~\ref{prop:threshold}
(Threshold Bayes rule)}
\label{proof:threshold}

Under the stated zero-one loss with asymmetric costs, the conditional expected
loss of each action given $\F_t$ is:
\begin{align*}
\E[L(1,D)\mid\F_t]
   &= c_{\mathrm{FP}}\,\Pbb(D=0\mid\F_t)
    = c_{\mathrm{FP}}(1-X_t), \\
\E[L(0,D)\mid\F_t]
   &= c_{\mathrm{FN}}\,\Pbb(D=1\mid\F_t)
    = c_{\mathrm{FN}}\,X_t.
\end{align*}
Treating is preferred ($a_t^*=1$) if and only if
\[
   c_{\mathrm{FP}}(1-X_t) < c_{\mathrm{FN}}\,X_t
   \;\iff\;
   (c_{\mathrm{FP}}+c_{\mathrm{FN}})\,X_t > c_{\mathrm{FP}}
   \;\iff\;
   X_t > \frac{c_{\mathrm{FP}}}{c_{\mathrm{FP}}+c_{\mathrm{FN}}}
      \;=:\; c^*.
\]
When $X_t=c^*$ both actions incur equal expected loss; the convention
$a_t^*=0$ is arbitrary.  The acting cost then evaluates to
$\ell(X_t)=c_{\mathrm{FN}}X_t\,\ind\{X_t\le c^*\}
+c_{\mathrm{FP}}(1-X_t)\,\ind\{X_t>c^*\}$,
which is piecewise linear and Lipschitz with constant
$L=\max(c_{\mathrm{FP}},c_{\mathrm{FN}})$.
This is a standard result in Bayes decision theory;
see \citet[Ch.~1--2]{berger1985}. \qed

\subsection{Proof of Theorem~\ref{thm:dp}
(Bellman recursion for optimal clinical stopping)}
\label{proof:dp}

We work with a finite horizon $T$ (the number of clinical stages).
Extensions to countably infinite horizons are standard;
see \citet{chow1971os} and \citet{peskir2006}.

\textbf{Boundary condition.}
At the terminal stage $T$ the clinician must act, so
$V_T=\ell(X_T)$.

\textbf{Well-posedness by backward induction.}
Suppose $V_{t+1}$ is $\F_{t+1}$-measurable and integrable.  Define
\[
   V_t := \min\!\left\{\ell(X_t),\; c_t+\E[V_{t+1}\mid\F_t]\right\}.
\]
Because $X_t\in[0,1]$ the acting cost $\ell(X_t)$ is bounded, and by the
induction hypothesis $\E[V_{t+1}\mid\F_t]$ is integrable, so $V_t$ is
$\F_t$-measurable and integrable.  Proceeding from $t=T-1$ down to $t=0$
defines $(V_t)_{t=0}^T$ uniquely.

\textbf{$V_t$ is a lower bound.}
We show by backward induction that for any $(\F_t)$-stopping time $\tau$ with
$\tau\le T$, the total expected cost satisfies
$\E\bigl[\ell(X_\tau)+\sum_{s=0}^{\tau-1}c_s\bigr]\ge V_0$.

The base case $t=T$ holds because $V_T=\ell(X_T)$.
Suppose the bound holds from stage $t+1$ onward, meaning for any stopping
time $\tau\ge t+1$,
$\E\bigl[\ell(X_\tau)+\sum_{s=t+1}^{\tau-1}c_s\bigm|\F_{t+1}\bigr]\ge V_{t+1}$
a.s.  At stage $t$:
\begin{itemize}[leftmargin=1.8em,itemsep=2pt]
\item If $\tau=t$ (stop and act): total cost from $t$ is $\ell(X_t)\ge V_t$,
  since $V_t\le\ell(X_t)$ by the min in the Bellman recursion.
\item If $\tau>t$ (continue): the total cost from $t$ equals
  $c_t + \ell(X_\tau)+\sum_{s=t+1}^{\tau-1}c_s$.
  Taking conditional expectation given $\F_t$ and applying the induction
  hypothesis gives
  $\E\bigl[c_t+\ell(X_\tau)+\sum_{s=t+1}^{\tau-1}c_s\bigm|\F_t\bigr]
  \ge c_t+\E[V_{t+1}\mid\F_t]\ge V_t$,
  where the last inequality uses the Bellman min.
\end{itemize}
In both cases total cost $\ge V_t$, confirming the lower bound at $t$.

\textbf{$\tau^*$ achieves the lower bound.}
Define $\tau^*=\min\{t: V_t=\ell(X_t)\}$.  Before $\tau^*$, by definition
$V_t<\ell(X_t)$, so the min is achieved by the continuation value:
$V_t=c_t+\E[V_{t+1}\mid\F_t]$.  At $\tau^*$, $V_{\tau^*}=\ell(X_{\tau^*})$.
Therefore
\[
   V_0 = \E\!\left[\ell(X_{\tau^*})+\sum_{s=0}^{\tau^*-1}c_s\right],
\]
so $\tau^*$ attains the lower bound and is optimal. \qed

\clearpage
\setcounter{section}{0}
\setcounter{subsection}{0}
\setcounter{subsubsection}{0}
\setcounter{figure}{0}
\setcounter{table}{0}
\renewcommand{\thesection}{S\arabic{section}}
\renewcommand{\thesubsection}{\thesection.\arabic{subsection}}
\renewcommand{\thesubsubsection}{\thesubsection.\arabic{subsubsection}}
\renewcommand{\thefigure}{S\arabic{figure}}
\renewcommand{\thetable}{S\arabic{table}}
\makeatletter
\renewcommand{\theHsection}{supp.\arabic{section}}
\renewcommand{\theHsubsection}{supp.\arabic{section}.\arabic{subsection}}
\renewcommand{\theHsubsubsection}{supp.\arabic{section}.\arabic{subsection}.\arabic{subsubsection}}
\renewcommand{\theHfigure}{supp.figure.\arabic{figure}}
\renewcommand{\theHtable}{supp.table.\arabic{table}}
\makeatother

\section{Additional Proof}
\label{sec:S1}

Proofs of Propositions~\ref{prop:martingale}, \ref{prop:revmart},
\ref{prop:projection_loss}, \ref{prop:threshold}, and
Theorem~\ref{thm:dp} are collected in Appendix~\ref{app:proofs} of this manuscript.  The remaining proof is given below.

\subsection{Proof of Proposition~\ref{prop:regret}
(Decision regret under posterior compression)}
\label{proof:regret}

By the Lipschitz assumption on $\ell$ with constant $L$,
\[
   \ell(\hat{Y}_t) - \ell(\hat{X}_t)
   \;\le\; \bigl|\ell(\hat{Y}_t)-\ell(\hat{X}_t)\bigr|
   \;\le\; L\,|\hat{X}_t-\hat{Y}_t| \quad\text{a.s.}
\]
Taking expectations and using linearity of $\E$:
\[
   \E[\ell(\hat{Y}_t)]-\E[\ell(\hat{X}_t)] \;\le\; L\,\E[|\hat{X}_t-\hat{Y}_t|].
\]
For the threshold loss of Proposition~\ref{prop:threshold}, $\ell$ is piecewise
linear with slopes $c_{\mathrm{FN}}$ on $(0,c^*)$ and $-c_{\mathrm{FP}}$ on
$(c^*,1)$, so $L=\max(c_{\mathrm{FP}},c_{\mathrm{FN}})$.  The bound is tight
(equality holds when $\hat{X}_t$ and $\hat{Y}_t$ straddle $c^*$ in the worst
case). \qed

\section{Additional Study Results}
\label{sec:S2}

The projection-loss analyses and martingale drift diagnostics omitted from the main text are collected here for Studies~1--4.

\subsection{Study 1: Breast Cancer Wisconsin}

\subsubsection{Projection loss}
\label{sec:bcw_projection}

We compressed the full $\F_3$ model into one-component (PCA1) and
three-component (PCA3) principal-component representations and measured
projection loss relative to the full posterior.  Table~\ref{tab:bcw_proj} and
Figure~\ref{fig:bcw_projection} show that aggressive compression (PCA1) incurs
substantial projection loss (mean Prob-MSE $= 0.042$), whereas moderate
compression (PCA3) retains much more of the decision-relevant structure (mean
Prob-MSE $= 0.014$).

A clinically important detail is that the PCA1 decision loss ($0.184$)
exceeds the decision loss at the basic pathology stage $\F_1$ alone ($0.154$;
Table~\ref{tab:bcw_decision}).  In other words, an overly simplified score can
be less useful for management than using a smaller set of raw clinical features.
A score that preserves high AUC is not automatically adequate if it distorts
which patients lie above or below the treatment threshold.

\begin{table}[htbp]
\centering
\caption{Projection and compression results (Study~1, 1{,}000 repetitions).
Prob-MSE vs.\ Full = mean$(\hat{X}_3-\hat{Y})^2$ averaged over the test set.}
\label{tab:bcw_proj}
\renewcommand{\arraystretch}{1.15}
\begin{tabular}{llrr}
\toprule
Model & Metric & Mean & SD \\
\midrule
Full (30 features) & AUC           & 0.9947 & 0.0041 \\
Full               & Brier         & 0.0203 & 0.0063 \\
Full               & Decision loss & 0.0860 & 0.0315 \\
Full               & Prob-MSE      & 0.0000 & 0.0000 \\
\addlinespace
PCA3 (3 components) & AUC          & 0.9903 & 0.0052 \\
PCA3                & Brier        & 0.0343 & 0.0079 \\
PCA3                & Decision loss & 0.1071 & 0.0324 \\
PCA3                & Prob-MSE     & 0.0143 & 0.0040 \\
\addlinespace
PCA1 (1 component) & AUC           & 0.9697 & 0.0109 \\
PCA1               & Brier         & 0.0637 & 0.0119 \\
PCA1               & Decision loss & 0.1836 & 0.0518 \\
PCA1               & Prob-MSE      & 0.0419 & 0.0087 \\
\bottomrule
\end{tabular}
\end{table}

\begin{figure}[htbp]
\centering
\begin{subfigure}[t]{0.48\textwidth}
\includegraphics[width=\textwidth]{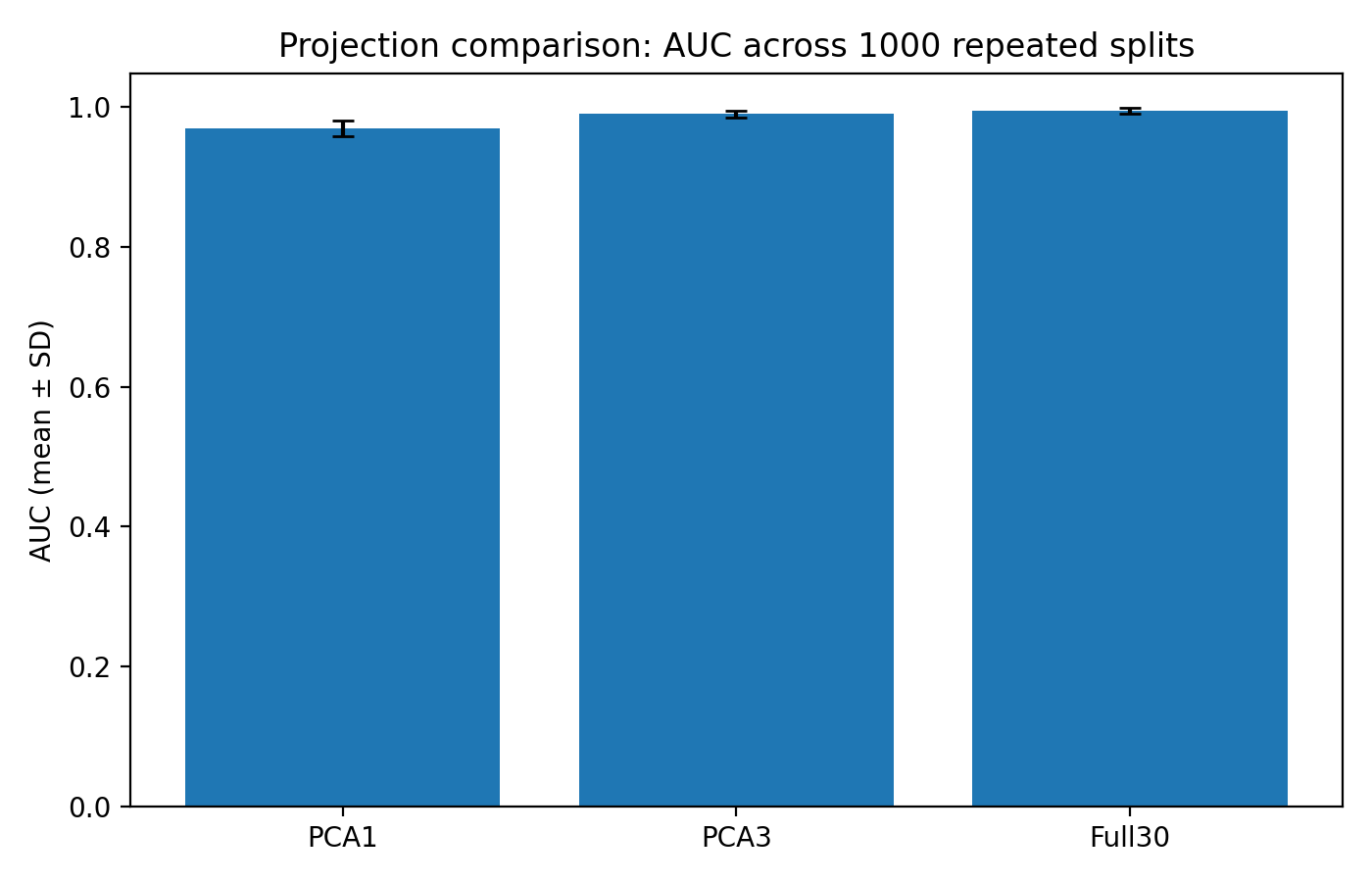}
\caption{AUC by compression level.}
\end{subfigure}
\hfill
\begin{subfigure}[t]{0.48\textwidth}
\includegraphics[width=\textwidth]{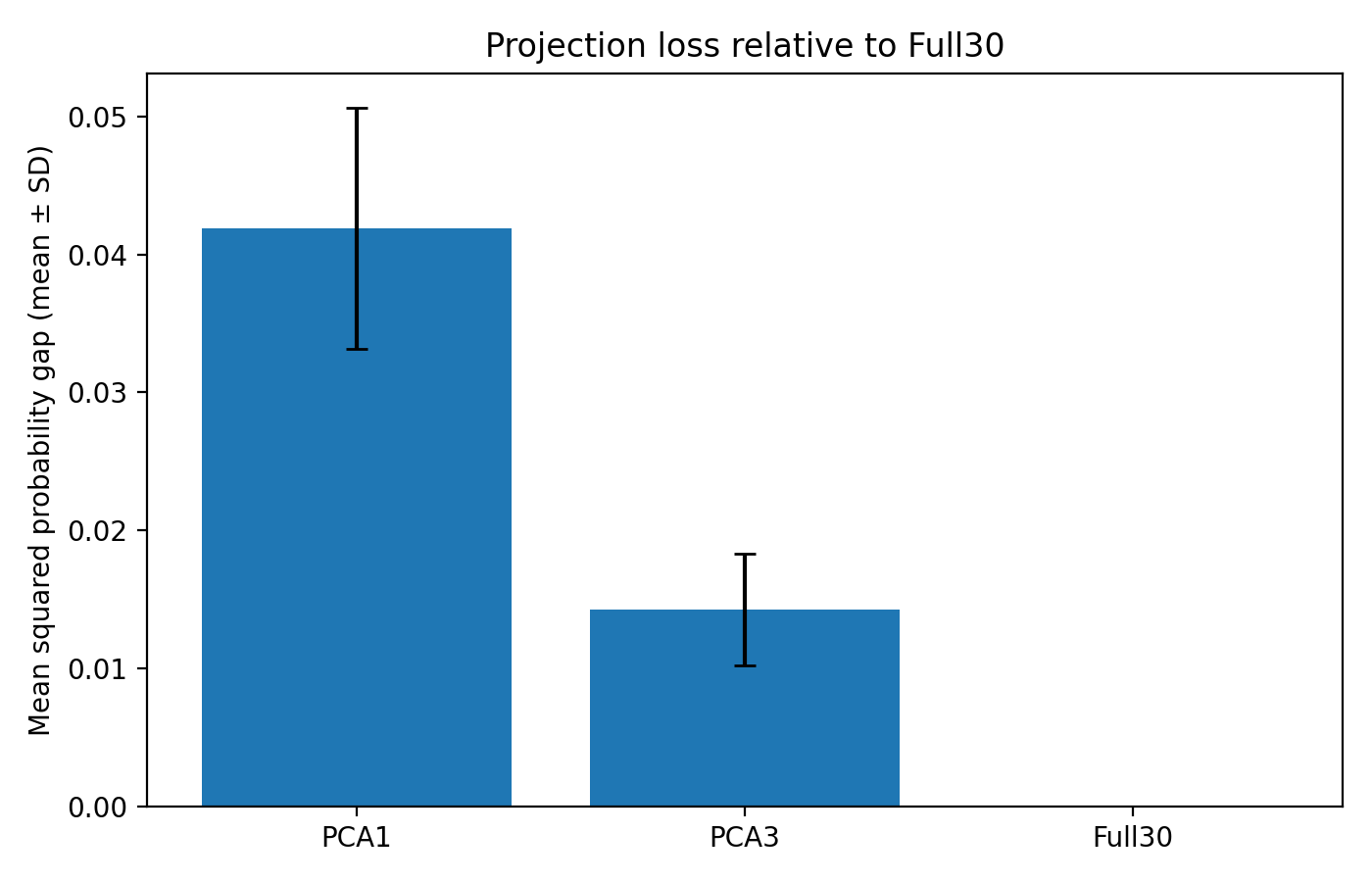}
\caption{Projection loss (Prob-MSE vs.\ full model).}
\end{subfigure}
\caption{Study~1: projection and compression.  PCA3 preserves most of the
full-information posterior structure; PCA1 incurs a substantial loss.}
\label{fig:bcw_projection}
\end{figure}

\subsubsection{Martingale drift diagnostic}
\label{sec:bcw_drift}

The weighted mean conditional drift was $-0.0025$ (SD $0.0083$) for the
$\F_1\to\F_2$ transition and $-0.0011$ (SD $0.0102$) for $\F_2\to\F_3$.
The weighted squared conditional drift was $0.00073$ and $0.00089$,
respectively.  These values are small relative to the risk scale $[0,1]$ and
are consistent with approximately coherent updating across stages.

In practical terms, the later pathology information appears to refine the risk
estimate rather than introduce erratic changes unrelated to the staged clinical
workup.

\begin{figure}[htbp]
\centering
\includegraphics[width=0.58\textwidth]{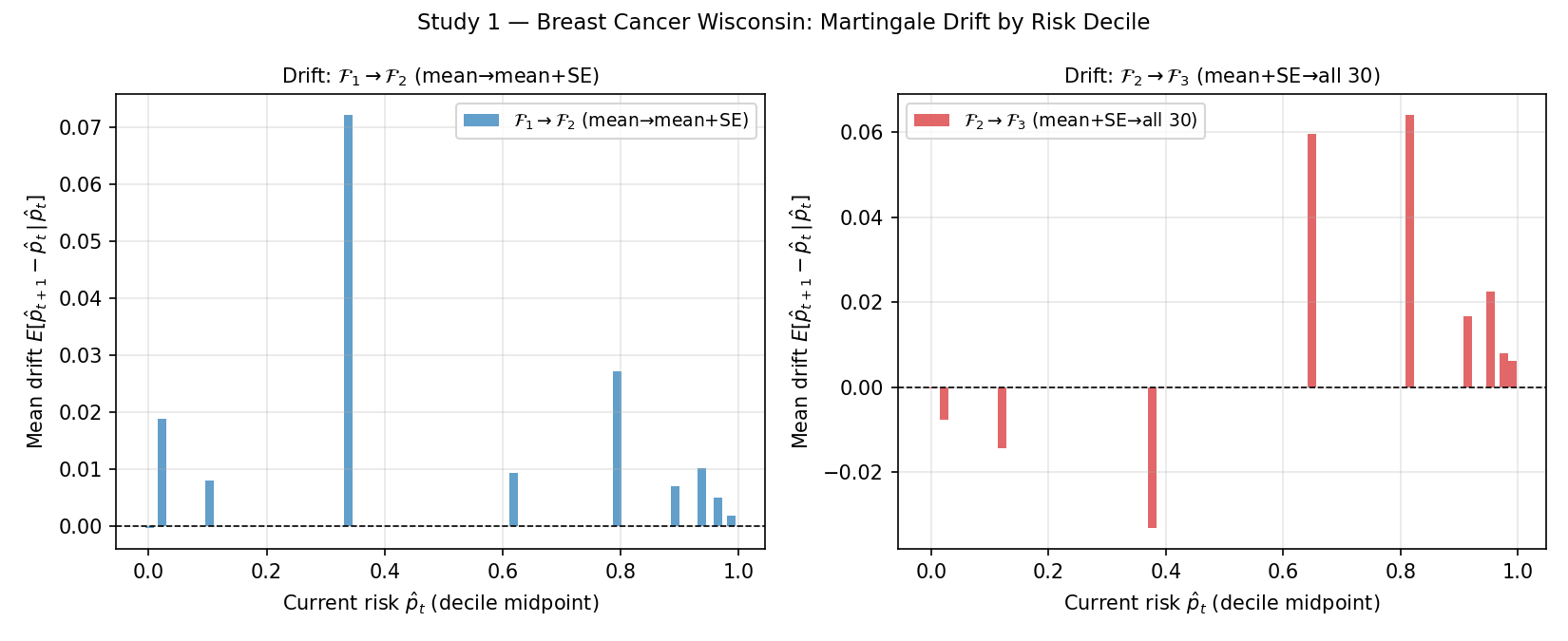}
\caption{Study~1: martingale drift diagnostic by risk decile.  Each bar shows
$E[\hat{p}_{t+1}-\hat{p}_t \,|\, \hat{p}_t \in \text{bin}_k]$ for ten
quantile bins of the current risk estimate.  The left panel ($\F_1\to\F_2$)
shows that positive drift concentrates in the intermediate-risk zone (around
0.3--0.4), reflecting that adding variability features raises estimates for
patients whose mean features placed them in a boundary region.  The right
panel ($\F_2\to\F_3$) shows both upward corrections in high-risk patients and
small downward corrections at the margin.  Global weighted mean drift is near
zero at both transitions, consistent with approximately coherent updating.}
\label{fig:bcw_drift}
\end{figure}


\subsection{Study 2: Cleveland Heart Disease}

\subsubsection{Projection loss}
\label{sec:cleveland_proj}

Table~\ref{tab:cleveland_proj} reports projection loss under PCA compression of
the full 13-feature model.

\begin{table}[htbp]
\centering
\caption{Projection loss (Cleveland Heart Disease, Study~2, 1{,}000 repetitions).
Prob-MSE vs.\ Full${}={}$mean $(\hat{X}_3 - \hat{Y})^2$ on test set.}
\label{tab:cleveland_proj}
\renewcommand{\arraystretch}{1.15}
\begin{tabular}{llrr}
\toprule
Model & Metric & Mean & SD \\
\midrule
Full (13 feat.) & AUC           & 0.898 & 0.028 \\
                & Brier         & 0.128 & 0.021 \\
                & Decision loss & 0.416 & 0.096 \\
                & Prob-MSE      & 0.000 & 0.000 \\
\addlinespace
PCA3 (3 comp.) & AUC            & 0.893 & 0.028 \\
               & Brier          & 0.131 & 0.019 \\
               & Decision loss  & 0.461 & 0.091 \\
               & Prob-MSE       & 0.019 & 0.006 \\
\addlinespace
PCA1 (1 comp.) & AUC            & 0.890 & 0.028 \\
               & Brier          & 0.134 & 0.019 \\
               & Decision loss  & 0.433 & 0.095 \\
               & Prob-MSE       & 0.024 & 0.006 \\
\bottomrule
\end{tabular}
\end{table}

Both compressed representations retain high AUC (above $0.890$), but the
decision loss for PCA3 ($0.461$) is higher than for PCA1 ($0.433$) despite its
lower probability MSE.  With only 13 features and a relatively small sample,
the three-component regression may be more vulnerable to estimation noise near
the treatment threshold.  The broader applied message is unchanged: a simplified
score can appear statistically similar to the full model while still changing
the clinical decision in nontrivial ways.

\begin{figure}[htbp]
\centering
\begin{subfigure}[t]{0.48\textwidth}
\includegraphics[width=\textwidth]{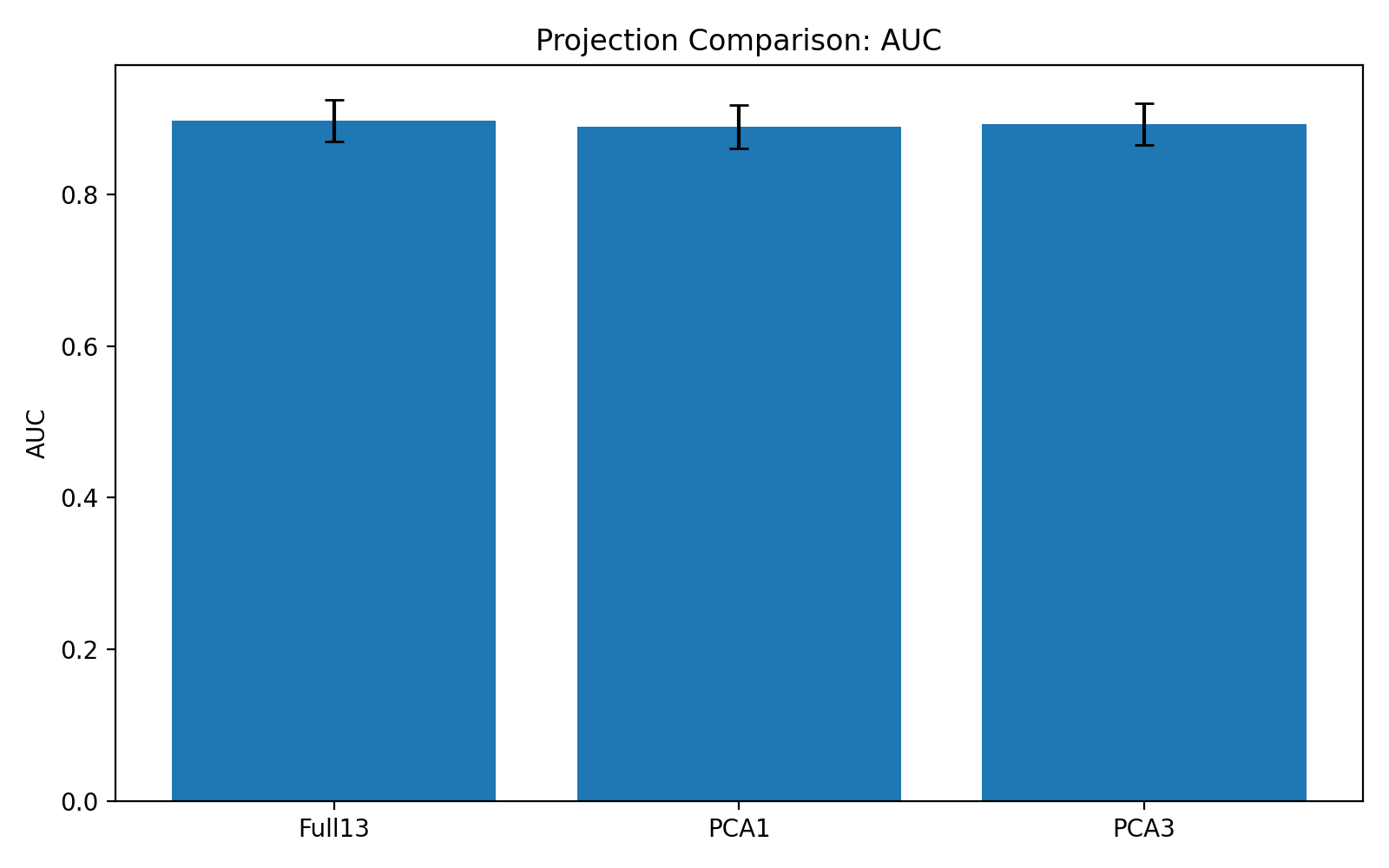}
\caption{AUC by compression level.}
\end{subfigure}
\hfill
\begin{subfigure}[t]{0.48\textwidth}
\includegraphics[width=\textwidth]{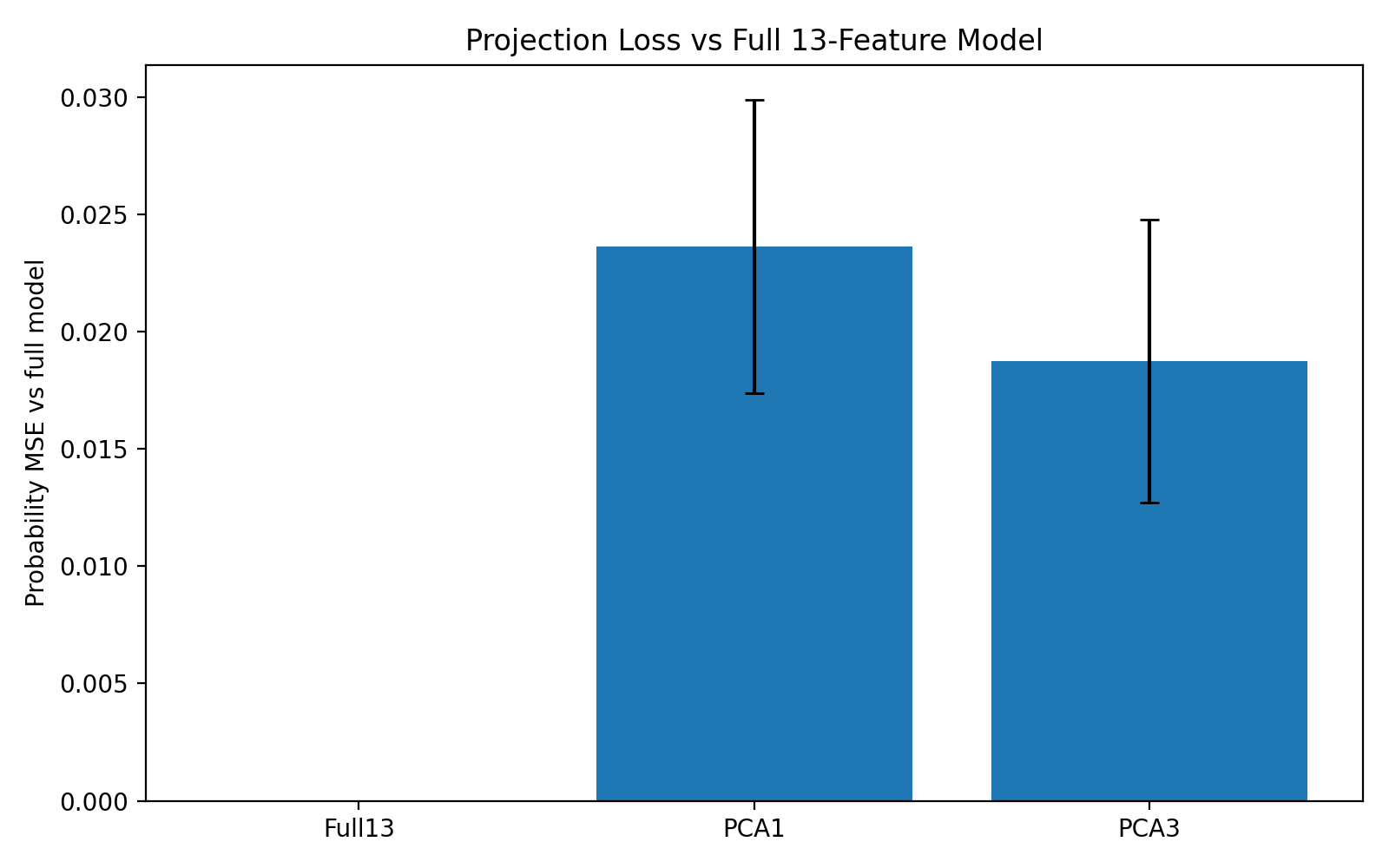}
\caption{Projection loss (Prob-MSE vs.\ full model).}
\end{subfigure}
\caption{Study~2: projection and compression.  PCA3 retains most of the AUC
but incurs higher decision loss than PCA1, illustrating that probability MSE
and threshold-based decision quality do not always align.}
\label{fig:cleveland_projection}
\end{figure}

\subsubsection{Martingale drift diagnostic}
\label{sec:cleveland_drift}

Figure~\ref{fig:cleveland_drift} shows the decile-level drift diagnostic.
The $\F_1\to\F_2$ panel reveals alternating positive and negative corrections
across the risk range, with larger magnitudes in the intermediate and
high-risk deciles.  The $\F_2\to\F_3$ panel shows a prominent positive spike
near risk 0.6: patients whose exercise-test posterior sits around 60\% tend to
have their estimate revised upward by the imaging/invasive features, while
lower- and higher-risk patients show smaller corrections.  The weighted mean
conditional drift across all bins is near zero at both transitions (global
$M_t \approx 0$), but the per-decile heterogeneity is nontrivial.  The
preferred stopping stage ($\F_2$) is driven by the cost-benefit trade-off
rather than by any systematic incoherence in the risk updates, but the drift
figure reveals that the direction of updating depends materially on where the
patient's risk estimate sits.

\begin{figure}[htbp]
\centering
\includegraphics[width=0.58\textwidth]{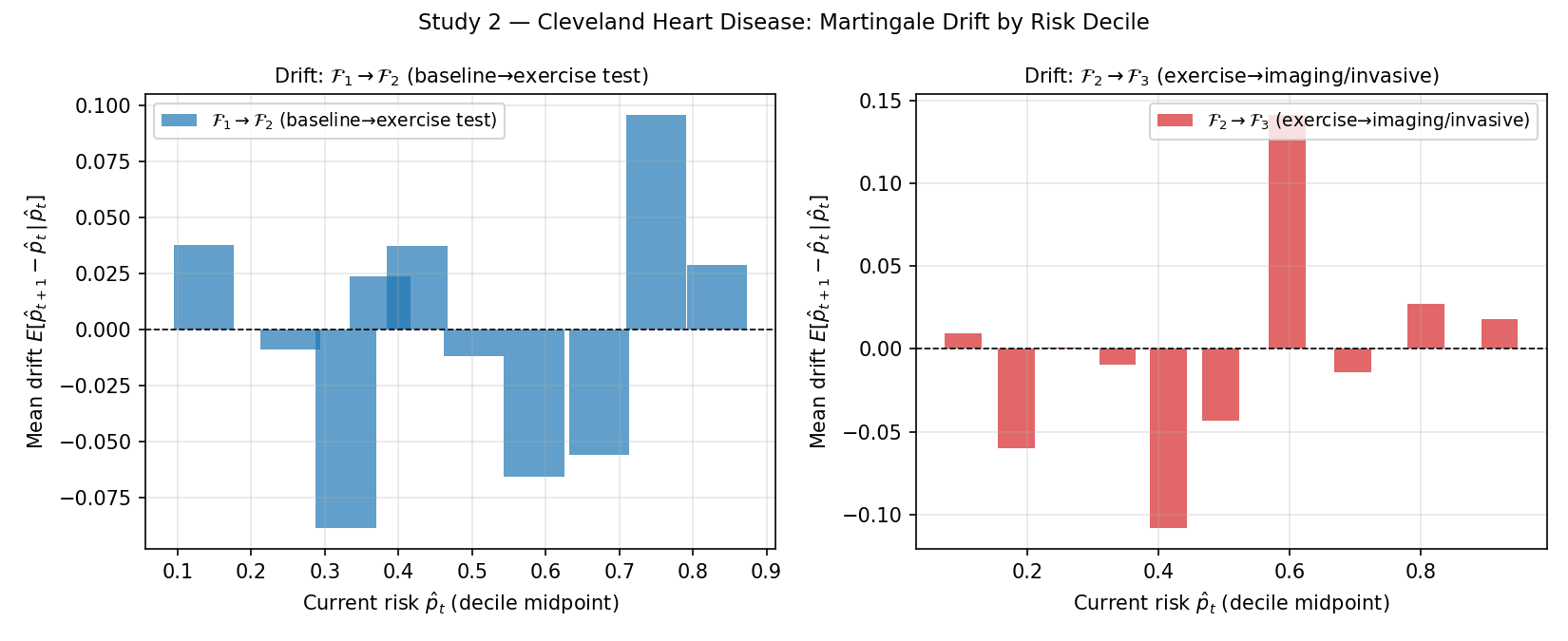}
\caption{Study~2: martingale drift diagnostic by risk decile.  Each bar shows
$E[\hat{p}_{t+1}-\hat{p}_t \,|\, \hat{p}_t \in \text{bin}_k]$.  The
$\F_1\to\F_2$ panel (left) shows alternating corrections across the risk
spectrum; the $\F_2\to\F_3$ panel (right) shows a pronounced positive spike
near risk 0.6, indicating that patients with intermediate posterior risk after
exercise testing tend to have their estimates revised upward by
imaging/invasive information.  The global weighted mean drift is near zero at
both transitions, but the heterogeneity by decile reveals that the martingale
approximation is more accurate at the risk extremes than in the middle range.}
\label{fig:cleveland_drift}
\end{figure}


\subsection{Study 3: Pima Diabetes}

\subsubsection{Projection loss}
\label{sec:diabetes_proj}

Table~\ref{tab:diabetes_proj} and Figure~\ref{fig:diabetes_projection} report
projection loss.  PCA compression of the $\F_3$ model shows that PCA3
(Prob-MSE $= 0.015$) retains substantially more posterior information than PCA1
(Prob-MSE $= 0.021$).  Decision loss under PCA3 ($0.449$) is nearly identical
to that of the full model ($0.448$), suggesting that a modestly simplified
summary can preserve most of the information needed for the threshold-based
decision, even when the full measurement set itself is not clearly superior to
earlier stages.

\begin{table}[htbp]
\centering
\caption{Projection loss (Pima Diabetes, Study~3, 1{,}000 repetitions).}
\label{tab:diabetes_proj}
\renewcommand{\arraystretch}{1.15}
\begin{tabular}{llrr}
\toprule
Model & Metric & Mean & SD \\
\midrule
Full (8 feat.)  & AUC           & 0.835 & 0.023 \\
                & Brier         & 0.158 & 0.011 \\
                & Decision loss & 0.448 & 0.051 \\
                & Prob-MSE      & 0.000 & 0.000 \\
\addlinespace
PCA3 (3 comp.) & AUC            & 0.815 & 0.023 \\
               & Brier          & 0.169 & 0.010 \\
               & Decision loss  & 0.449 & 0.042 \\
               & Prob-MSE       & 0.015 & 0.003 \\
\addlinespace
PCA1 (1 comp.) & AUC            & 0.797 & 0.024 \\
               & Brier          & 0.175 & 0.010 \\
               & Decision loss  & 0.459 & 0.042 \\
               & Prob-MSE       & 0.021 & 0.003 \\
\bottomrule
\end{tabular}
\end{table}

\begin{figure}[htbp]
\centering
\begin{subfigure}[t]{0.48\textwidth}
\includegraphics[width=\textwidth]{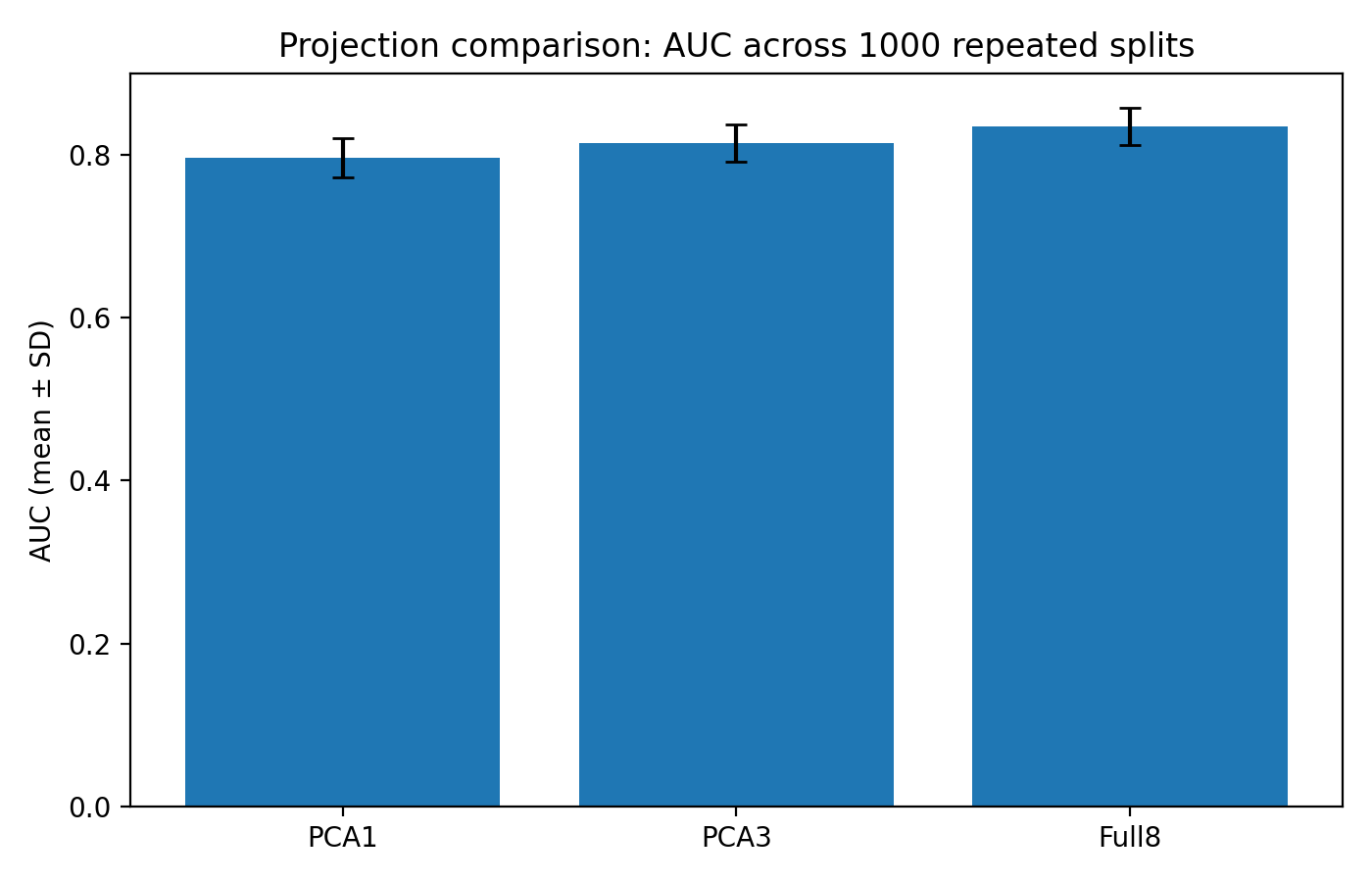}
\caption{AUC by compression level.}
\end{subfigure}
\hfill
\begin{subfigure}[t]{0.48\textwidth}
\includegraphics[width=\textwidth]{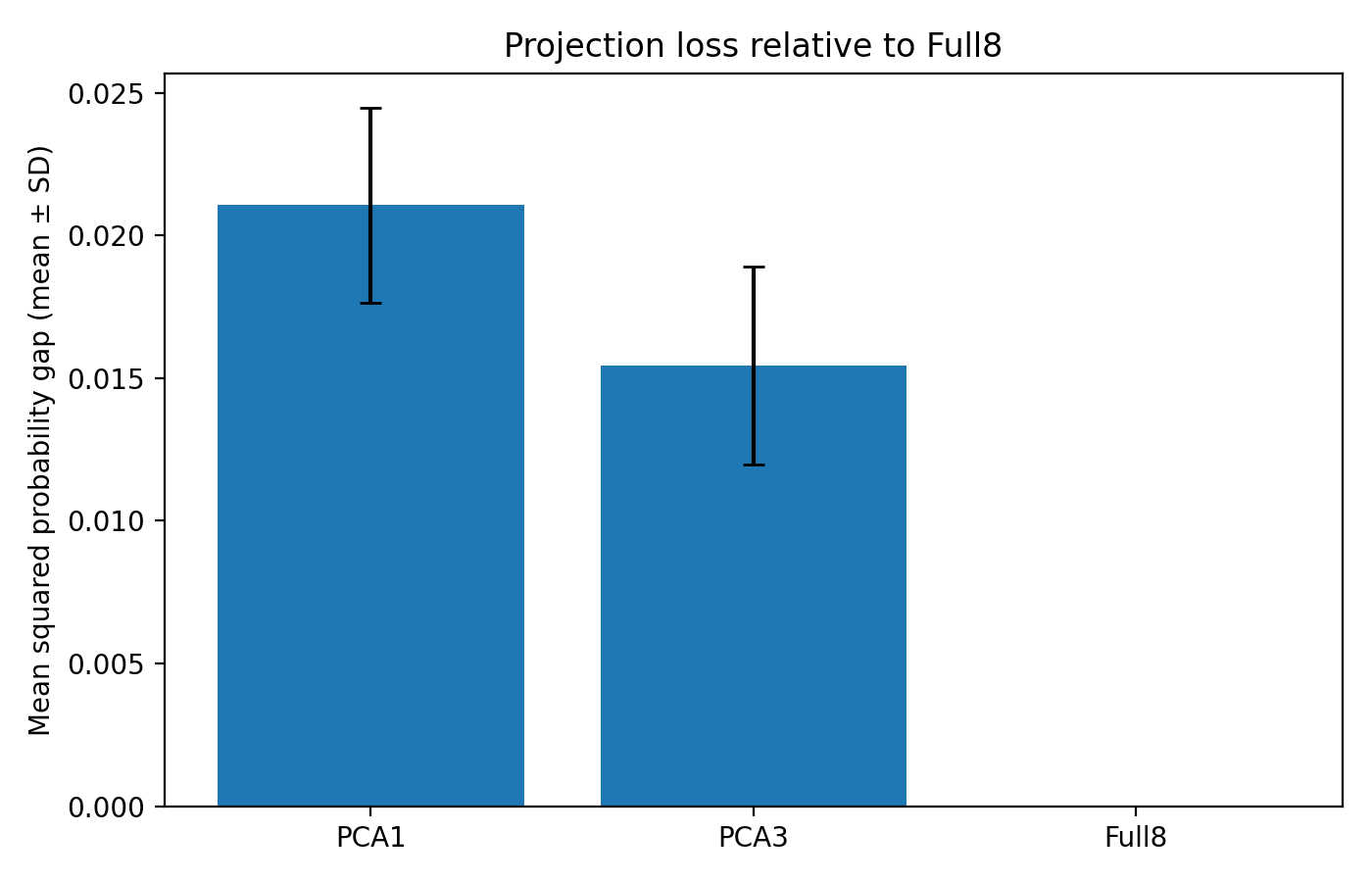}
\caption{Projection loss (Prob-MSE vs.\ full model).}
\end{subfigure}
\caption{Study~3: projection and compression.  PCA3 incurs lower probability
MSE than PCA1 and achieves nearly the same decision loss as the full model,
suggesting that a three-component summary preserves most of the
threshold-relevant information in this dataset.}
\label{fig:diabetes_projection}
\end{figure}

\subsubsection{Martingale drift diagnostic}
\label{sec:diabetes_drift}

The weighted mean conditional drift was $-4.4\times 10^{-5}$ (SD $0.006$) for
the $\F_1\to\F_2$ transition and $+5.9\times 10^{-5}$ (SD $0.002$) for
$\F_2\to\F_3$.  Both values are negligible, indicating approximately coherent
risk updating despite the non-monotone AUC pattern.  In practical terms, the
problem here is not erratic stage ordering but limited incremental value of the
later measurements for the chosen clinical decision.

\begin{figure}[htbp]
\centering
\includegraphics[width=0.98\textwidth]{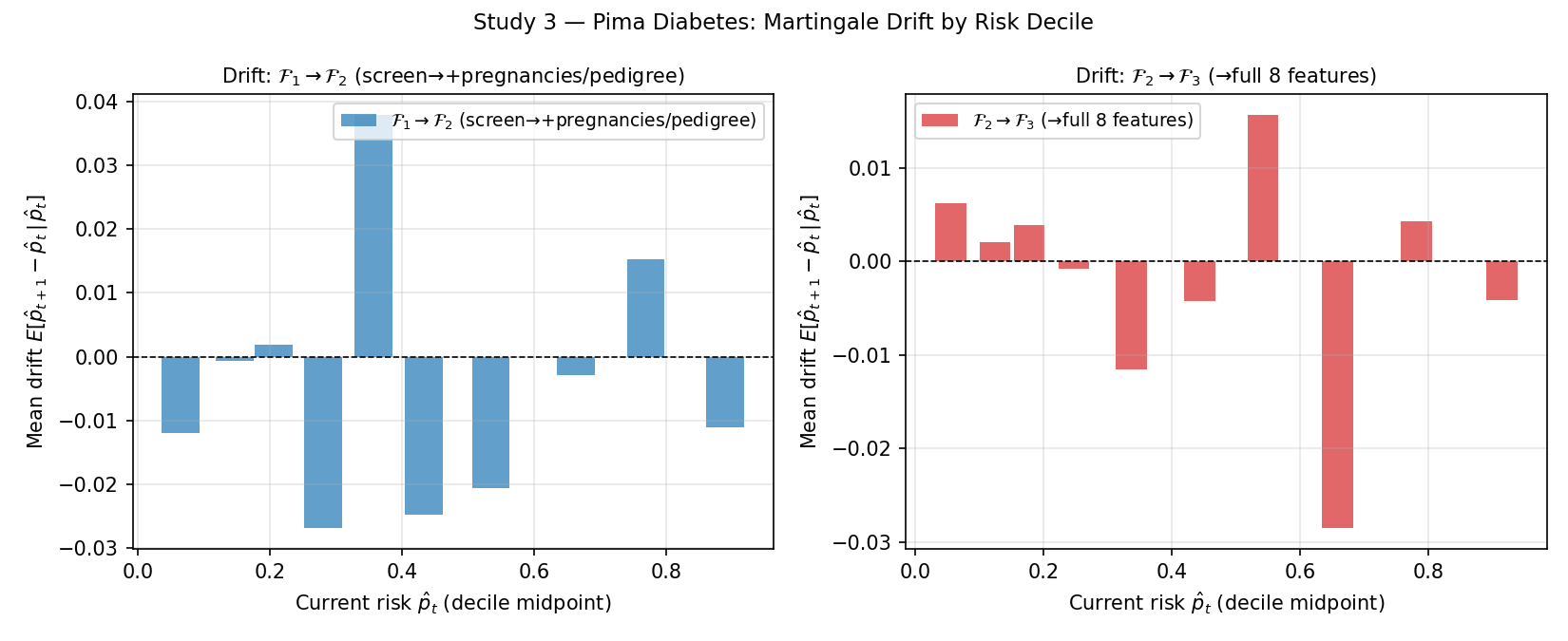}
\caption{Study~3: martingale drift diagnostic by risk decile.  Each bar shows
$E[\hat{p}_{t+1}-\hat{p}_t \,|\, \hat{p}_t \in \text{bin}_k]$.  Both panels
show oscillating positive/negative corrections of modest magnitude across
deciles, with the largest individual bars below $|0.04|$ in absolute value.
The pattern is consistent with the non-monotone AUC: the later features
neither systematically raise nor lower risk estimates in a directional way,
confirming that the issue is limited incremental information content rather
than incoherent updating.}
\label{fig:diabetes_drift}
\end{figure}



\subsection{Study 4: eICU Demo}

\subsubsection{Projection loss}
\label{sec:eicu_proj}

Table~\ref{tab:eicu_proj} reports the stagewise projection loss.  The
$\F_1\to\F_2$ projection loss ($0.0182$, SD $0.0033$) and $\F_2\to\F_3$
projection loss ($0.0209$, SD $0.0034$) are both modest, indicating that each
stage transition revises risk estimates by a small amount on average.  The
$\F_2\to\F_3$ loss is slightly larger than $\F_1\to\F_2$, consistent with the
$\F_3$ stage providing the more influential revision (the one that actually
changes the decision for a meaningful fraction of patients).

\begin{table}[htbp]
\centering
\caption{Stagewise projection loss (eICU Demo, Study~4, 1{,}000 repetitions).
$\widehat{\mathrm{PL}}_{t} = n_{\mathrm{te}}^{-1}\sum_i(\hat{p}_{t,i} -
\hat{p}_{t+1,i})^2$, averaged over the test set.}
\label{tab:eicu_proj}
\renewcommand{\arraystretch}{1.15}
\begin{tabular}{lrr}
\toprule
Transition & Projection loss (mean) & SD \\
\midrule
$\F_1 \to \F_2$ & 0.0182 & 0.0033 \\
$\F_2 \to \F_3$ & 0.0209 & 0.0034 \\
\bottomrule
\end{tabular}
\end{table}

\begin{figure}[htbp]
\centering
\includegraphics[width=0.58\textwidth]{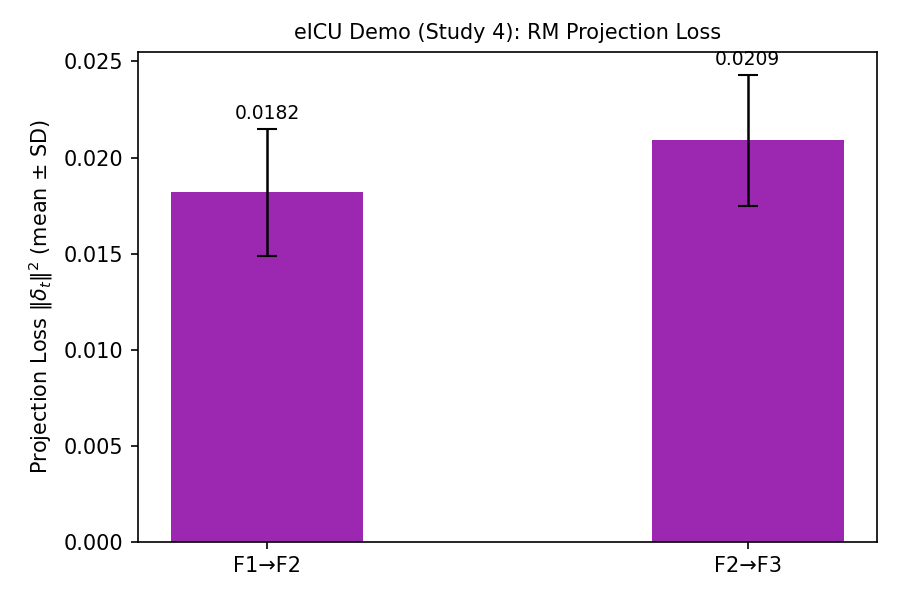}
\caption{eICU Demo (Study~4): reverse-martingale projection loss across stage
transitions.  The $\F_2\to\F_3$ transition has slightly higher projection loss,
consistent with the larger decision improvement at that transition.}
\label{fig:eicu_proj}
\end{figure}

\subsubsection{Martingale drift diagnostic}
\label{sec:eicu_drift}

The decile-level drift diagnostic is shown in Figure~\ref{fig:eicu_drift}.
Both the $\F_1\to\F_2$ and $\F_2\to\F_3$ panels show a consistent pattern of
\emph{negative} drift: $\hat{p}_{t+1} < \hat{p}_t$ on average within each
risk decile.  The weighted mean drift is $M = -0.035$ for $\F_1\to\F_2$ and
$M = -0.030$ for $\F_2\to\F_3$.

This systematic downward revision is interpretively meaningful.  The $\F_1$
model, fitted to vital signs alone in a class-imbalanced setting, cannot
identify low-risk patients and assigns probabilities above $c^*$ to nearly
everyone.  Adding laboratory values at $\F_2$ provides some refinement, and
adding demographic features at $\F_3$ enables a more substantial downward
revision of risk for patients who are younger, lighter, or admitted to less
acute ICU units.  The drift is therefore not evidence of incoherence---it
reflects that early-stage models systematically overestimate risk relative to
later, more informative stages, which is the expected behaviour when later
features specifically help identify the low-risk tail.  The weighted squared
drift ($S_t = 0.0017$ and $0.0013$) confirms that the revisions are
consistent rather than erratic.

\begin{figure}[htbp]
\centering
\includegraphics[width=0.95\textwidth]{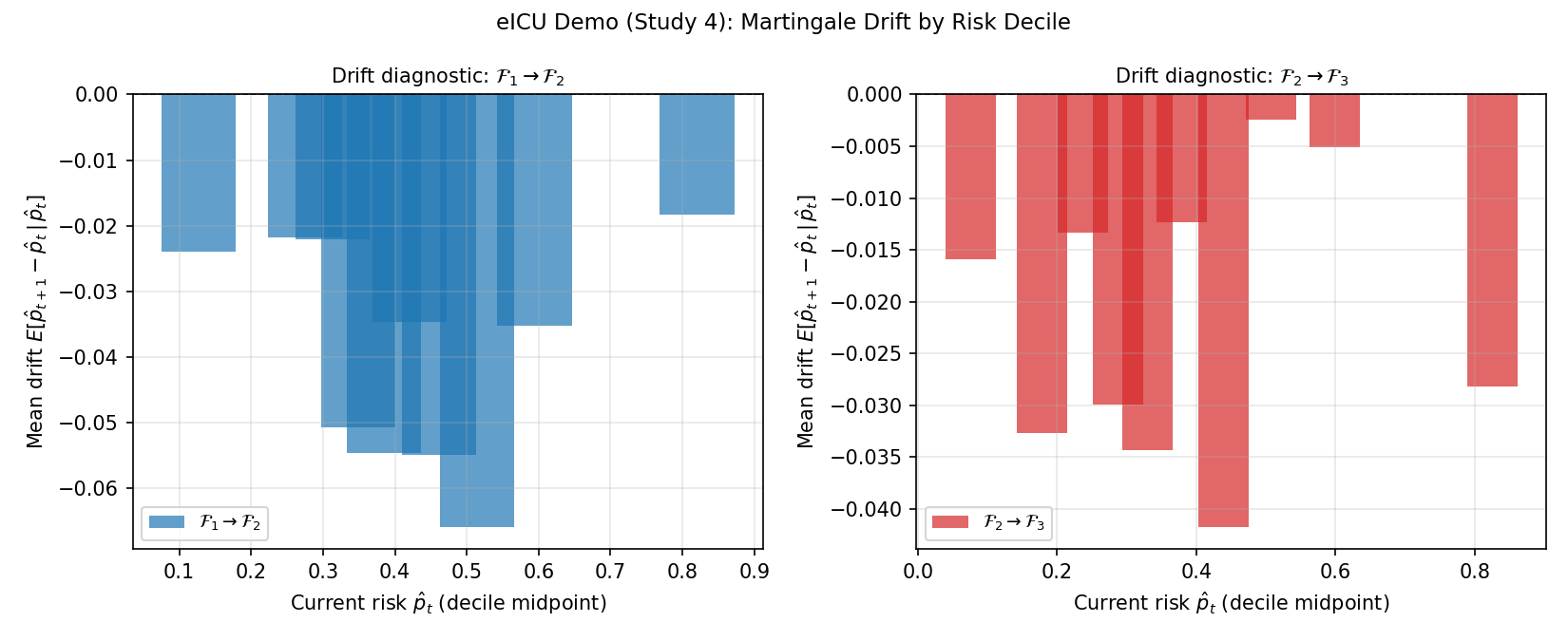}
\caption{eICU Demo (Study~4): martingale drift diagnostic by risk decile.
Each bar shows $E[\hat{p}_{t+1}-\hat{p}_t \,|\, \hat{p}_t \in \text{bin}_k]$
for ten quantile bins.  Both transitions show consistent negative drift across
all deciles, reflecting that adding laboratory values and especially
demographic features systematically reduces predicted risk.  This is
interpretively coherent: the early-stage model flags nearly everyone under the
low threshold $c^*\approx 0.091$, and later stages provide information that
correctly identifies lower-risk patients.}
\label{fig:eicu_drift}
\end{figure}



\section{Code and Reproducibility}
\label{sec:S3}

All analysis scripts are in the \texttt{scripts/} subfolder:

\begin{center}
\renewcommand{\arraystretch}{1.2}
\begin{tabular}{ll}
\toprule
Script & Purpose \\
\midrule
\texttt{run\_bcw\_analysis.py}      & Study~1 (Breast Cancer Wisconsin) \\
\texttt{run\_heart\_analysis.py}    & Study~2 (Cleveland Heart Disease) \\
\texttt{run\_diabetes\_analysis.py} & Study~3 (Pima Diabetes) \\
\texttt{run\_eicu\_analysis.py}     & Study~4 (eICU Demo, 1{,}000 reps) \\
\bottomrule
\end{tabular}
\end{center}

Dependencies: \texttt{numpy}, \texttt{pandas}, \texttt{scikit-learn},
\texttt{matplotlib}.
Install via: \texttt{pip install numpy pandas scikit-learn matplotlib}.

\section{Sensitivity Analysis: Stage-Cost Robustness}
\label{sec:S4}

\subsection{Stage-cost sensitivity for Studies 2 and 4}

The main text uses a single illustrative stage-cost schedule for each study.
Table~\ref{tab:sensitivity} reports the preferred stopping stage for
Study~2 (Cleveland Heart Disease) and Study~4 (eICU) under a range of
alternative cumulative stage-cost schedules $(0, c_2, c_3)$, holding the
loss ratio and decision losses fixed at their reported values.

The stopping conclusion for Study~2 ($\F_2$ preferred) is robust to
moderate cost increases, switching to $\F_1$ only when the cost of
moving to $\F_2$ alone exceeds $\approx 0.062$ --- larger than
any plausible exercise-test burden relative to the decision scale.
The stopping conclusion for Study~4 ($\F_3$ preferred) is even more
robust: the large decision-loss reduction from $\F_2$ to $\F_3$
($\Delta\hat{\ell} = 0.095$) means that the stage-$3$ cost must exceed
$0.10$ before $\F_1$ becomes preferred.

\subsection{eICU alternative stage ordering}

As noted in Sections~\ref{sec:datasets} and~\ref{sec:disc_limits}, the
analytic staging of Study~4 places demographics at $\F_3$ for predictive
convenience.  Under the alternative ordering (demographics$+$vitals at
$\F_1$, $+$labs at $\F_2$, full at $\F_3$), AUC at $\F_1$ rises to
approximately $0.74$--$0.75$ and specificity improves substantially; the
AUC at $\F_3$ is unchanged.  The preferred stopping stage under this
ordering is still $\F_3$: the full feature set minimises total expected
cost, and the first-stage decision quality, while better than in the
original staging, does not reach the level of the full model.  This
confirms that the main qualitative finding for Study~4 is not an artefact
of the analytic stage convention.
\begin{table}[ht]
\centering
\caption{Stage-cost sensitivity: preferred stopping stage under alternative
cumulative cost schedules $(0, c_2, c_3)$.  Decision losses are held at
their reported values (Study~2: $\hat{\ell} = 0.486, 0.424, 0.416$;
Study~4: $\hat{\ell} = 0.918, 0.911, 0.816$).
$\dagger$ marks the minimum total cost at each row.
Baseline schedule marked with $*$.}
\label{tab:sensitivity}
\renewcommand{\arraystretch}{1.2}
\begin{tabular}{lrrrrrl}
\toprule
Study & $c_2$ & $c_3$ & Total$(\F_1)$ & Total$(\F_2)$ & Total$(\F_3)$ & Preferred \\
\midrule
\multicolumn{7}{l}{\textit{Study~2 --- Cleveland Heart Disease
  ($\hat{\ell}_1=0.486,\;\hat{\ell}_2=0.424,\;\hat{\ell}_3=0.416$)}} \\
  & 0.01 & 0.03 & 0.486 & 0.434$^\dagger$ & 0.446 & $\F_2$ \\
  & 0.02 & 0.06$^*$ & 0.486 & 0.444$^\dagger$ & 0.476 & $\F_2$ \\
  & 0.04 & 0.08 & 0.486 & 0.464$^\dagger$ & 0.496 & $\F_2$ \\
  & 0.06 & 0.12 & 0.486 & 0.484$^\dagger$ & 0.536 & $\F_2$ \\
  & 0.07 & 0.14 & 0.486$^\dagger$ & 0.494 & 0.556 & $\F_1$ \\
\addlinespace
\multicolumn{7}{l}{\textit{Study~4 --- eICU Demo
  ($\hat{\ell}_1=0.918,\;\hat{\ell}_2=0.911,\;\hat{\ell}_3=0.816$)}} \\
  & 0.005 & 0.01 & 0.918 & 0.916 & 0.826$^\dagger$ & $\F_3$ \\
  & 0.010 & 0.03$^*$ & 0.918 & 0.921 & 0.846$^\dagger$ & $\F_3$ \\
  & 0.010 & 0.05 & 0.918 & 0.921 & 0.866$^\dagger$ & $\F_3$ \\
  & 0.010 & 0.08 & 0.918 & 0.921 & 0.896$^\dagger$ & $\F_3$ \\
  & 0.010 & 0.10 & 0.918 & 0.921 & 0.916$^\dagger$ & $\F_3$ \\
  & 0.010 & 0.11 & 0.918$^\dagger$ & 0.921 & 0.926 & $\F_1$ \\
\bottomrule
\end{tabular}
\end{table}

\section{Calibration Assessment}
\label{sec:S5}

Calibration of the stagewise risk estimates was assessed by fitting a
logistic recalibration model on the test-set predictions: the observed
binary outcome was regressed on the log-odds of the stage-$t$ predicted
probability across each of the 1{,}000 repeated splits, and the mean
calibration slope and intercept were recorded.  A well-calibrated model
has slope $\approx 1$ and intercept $\approx 0$.

\begin{table}[htbp]
\centering
\caption{Stagewise calibration: mean (SD) of calibration slope and intercept
over 1{,}000 repeated 70/30 stratified splits, with 2.5/97.5 percentile
winsorisation before averaging.  A slope of 1 and intercept of 0 indicate
perfect calibration.
$^*$~Near-perfect separation at $\F_3$ (AUC $\approx$ 0.995) renders logistic
recalibration unreliable; slope estimates are inflated and have high variance.
$^\dagger$~The eICU demo has 8.6\% in-hospital mortality; severe class
imbalance combined with the low threshold $c^*\approx 0.091$ produces
near-separation in early stages, causing logistic recalibration slopes and
intercepts to be unstable with high variance.  Results are reported for
completeness.}
\label{tab:calibration}
\renewcommand{\arraystretch}{1.15}
\begin{tabular}{llrrrr}
\toprule
Study & Stage & Slope (mean) & Slope (SD) & Intercept (mean) & Intercept (SD) \\
\midrule
BCW (Study~1)
  & $\F_1$ & 1.22 & 0.23 & 0.03 & 0.42 \\
  & $\F_2$ & 1.23 & 0.25 & 0.10 & 0.45 \\
  & $\F_3$$^*$ & 1.93 & 1.93 & 0.27 & 0.77 \\
\addlinespace
Cleveland (Study~2)
  & $\F_1$ & 0.98 & 0.27 & 0.01 & 0.16 \\
  & $\F_2$ & 0.89 & 0.21 & $-0.01$ & 0.20 \\
  & $\F_3$ & 0.90 & 0.21 & $-0.01$ & 0.26 \\
\addlinespace
Pima (Study~3)
  & $\F_1$ & 1.02 & 0.18 & 0.01 & 0.14 \\
  & $\F_2$ & 1.00 & 0.18 & 0.00 & 0.14 \\
  & $\F_3$ & 0.97 & 0.17 & $-0.01$ & 0.14 \\
\addlinespace
eICU (Study~4)$^\dagger$
  & $\F_1$ & 9.55 & 26.52 & 1.95 & 19.81 \\
  & $\F_2$ & 8.89 & 16.69 & 1.74 & 20.23 \\
  & $\F_3$ & 9.28 & 20.08 & 4.26 & 14.92 \\
\bottomrule
\end{tabular}
\end{table}

Calibration slopes are expected to be close to 1 for $\ell_2$-regularised
logistic regression on moderately sized samples, as regularisation
shrinks predicted probabilities toward the prevalence and the recalibration
regression corrects for this shrinkage on the test set.  Any systematic
miscalibration across stages would threaten the martingale interpretation
(since the framework requires $\hat{X}_t \approx \Pbb(D=1\mid\F_t)$) and
should be noted.  For Studies~1--3, slopes are close to 1 and intercepts
close to 0, indicating adequate calibration of the $\ell_2$-regularised
logistic regression models.  The exceptions ($\F_3$ of Study~1 and
Study~4) are explained by the footnotes to Table~\ref{tab:calibration}.

\end{document}